\documentclass[]{interact}

\usepackage{amssymb, amsmath, amsthm, graphicx, epstopdf, booktabs, xcolor}
\usepackage{epstopdf}
\usepackage{subfigure}
\usepackage{bm} 
\usepackage{mathrsfs} 
\usepackage{commath} 

\usepackage{hyperref}

\usepackage{natbib}
\bibpunct[, ]{(}{)}{;}{a}{}{,}

\allowdisplaybreaks[4]

\theoremstyle{plain}
\newtheorem{thm}{Theorem}[section]
\newtheorem{lem}[thm]{Lemma}
\newtheorem{cor}[thm]{Corollary}
\newtheorem{prop}[thm]{Proposition}

\theoremstyle{definition}

\newtheorem{asp}[thm]{Assumption}

\theoremstyle{remark}
\newtheorem{rem}{Remark}

\newcommand{\pder}[2][]{\frac{\partial#1}{\partial#2}}

\newcommand{\calA}{\mathcal{A}}

\newcommand{\calC}{\mathcal{C}}
\newcommand{\calF}{\mathcal{F}}
\newcommand{\calL}{\mathcal{L}}
\newcommand{\calM}{\mathcal{M}}
\newcommand{\calS}{\mathcal{S}}
\newcommand{\frakM}{\mathfrak{M}}

\newcommand{\E}{\mathbb{E}}				
\newcommand{\prob}{\mathbb{P}}		
\newcommand{\Q}{\mathbb{Q}}

\newcommand{\vm}[1]{\mathbf{#1}}

\begin{document}


\title{Representation of Exchange Option Prices under Stochastic Volatility Jump-Diffusion Dynamics\thanks{This is the preprint of the article of the same title published by Taylor \& Francis in \emph{Quantitative Finance} \textbf{20}(2), available online at: \url{https://doi.org/10.1080/14697688.2019.1655785}.}}

\author{
\name{Gerald H. L. Cheang\textsuperscript{a}\thanks{E-mail: Gerald.Cheang@unisa.edu.au; ORCiD: \url{https://orcid.org/0000-0003-3786-0285}} and Len Patrick Dominic M. Garces\textsuperscript{a,b}\thanks{Corresponding Author; E-mail: len\_patrick\_dominic.garces@mymail.unisa.edu.au; ORCiD: \url{https://orcid.org/0000-0002-2737-7348
}}}
\affil{\textsuperscript{a}Centre for Industrial and Applied Mathematics, School of Information Technology and Mathematical Sciences, University of South Australia, Mawson Lakes SA 5095, Australia; \textsuperscript{b}Department of Mathematics, School of Science and Engineering, Ateneo de Manila University, Quezon City 1108, Metro Manila, Philippines}
}

\maketitle

\begin{abstract}
In this article, we provide representations of European and American exchange option prices under stochastic volatility jump-diffusion (SVJD) dynamics following models by \citet{Merton-1976}, \citet{Heston-1993}, and \citet{Bates-1996}. A Radon-Nikod\'ym derivative process is also introduced to facilitate the shift from the objective market measure to other equivalent probability measures, including the equivalent martingale measure. Under the equivalent martingale measure, we derive the integro-partial differential equation that characterizes the exchange option prices. We also derive representations of the European exchange option price using the change-of-num\'eraire technique proposed by \citet{Geman-1995} and the Fourier inversion formula derived by \citet{CaldanaFusai-2013}, and show that these two representations are comparable. Lastly, we show that the American exchange option price can be decomposed into the price of the European exchange option and an early exercise premium.\end{abstract}


\section{Introduction}

An exchange option is a contract that grants the holder the right, but not the obligation, to exchange one risky asset for another. For example, if time $t$ asset prices are denoted by $S_{1,t}$ and $S_{2,t}$, then the payoff of the European exchange option with maturity $T$ is given by $(S_{1,T}-S_{2,T})^+$, where $x^+ = \max\{x,0\}$. The exchange option is a special case of the spread option, which is an option written on the difference of the prices of two assets. For example, the European call spread option with strike price $K$ has terminal payoff $(S_{1,T}-S_{2,T}-K)^+$, and so the exchange option can be seen as a call spread option with zero strike price. In the \citet{BlackScholes-1973} framework, the price of the European exchange option is given by the celebrated \citet{Margrabe-1978} formula. However, \citet{Carmona-2003} established that such a closed-form equation is unavailable for general non-zero strike European spread options even in the Black-Scholes model.

In recent years, literature in option pricing has moved beyond the classical model. Empirical studies have shown that the classical Black-Scholes assumption that asset prices are log-normally distributed is insufficient to capture pertinent features of asset returns such as heavy tails, volatility clustering, and implied volatility smiles and skews \citep{Cont-2001, ContTankov-2004, Kou-2008}. In light of the limitations of the Black-Scholes framework, alternative asset price models have been proposed to provide more accurate characterizations of asset returns. Prominent examples of these alternative models are jump-diffusion models \citep{Merton-1976, NaikLee-1990, Pham-1997, Kou-2002}, stochastic volatility models \citep{HullWhite-1987, SteinStein-1991, Heston-1993}, and combinations of stochastic volatility and jump-diffusion models \citep{Bates-1996, Bakshi-1997, Scott-1997}. 

A practical consequence, however, of the use of alternative asset price models is that option prices are no longer available in a form as elegant as the original Black-Scholes formula. Indeed, the option pricing formulas obtained by \citet{Heston-1993} and \citet{Bates-1996} are expressed in ``semi-closed'' forms in terms of the characteristic function of log-prices; in the \citet{Merton-1976} jump-diffusion setting, \citet{CheangChiarella-2012} obtain an infinite series representation of European option prices with Poisson-probability weights. A popular approach that can be efficiently implemented using computing software is the fast Fourier transform (FFT) approach proposed by \citet{CarrMadan-1999} and \citet{Lewis-2001}.

The complexity of the pricing problem is further exacerbated when considering early-exercise and American options, where one has to also account for the early exercise boundary. In the pure diffusion setting, \citet{Kim-1990}, \citet{Jacka-1991}, and \citet{Jamshidian-1992} were able to show that the price of the American option on a single stock decomposes into the sum of the price of the corresponding European option and a quantity that is commonly interpreted as the early exercise premium. \citet{Pham-1997} and \citet{Gukhal-2001} were able to derive a similar representation in the jump-diffusion setting, noting that the early exercise premium is heavily affected by the possibility of jumps in asset prices. \citet{Pham-1997} and \citet{Touzi-1999}, in the jump-diffusion and the stochastic volatility frameworks, respectively, provide an analysis of the American option price with respect to the early exercise boundary. \citet{CheangChiarellaZiogas-2013} show that a similar decomposition holds under a stochastic volatility jump-diffusion model for the underlying asset. 

Parallel to the developments in single-asset option pricing under alternative price processes, the valuation of the European exchange option has also since then been conducted under jump-diffusion models and stochastic volatility models. \citet{Jamshidian-2007} considered the pricing of European exchange options and constructing hedging portfolios when asset prices are driven by pure-diffusion processes with deterministic volatility and when they are modelled using exponential Poisson processes. \citet{Antonelli-2010}, \citet{Alos-2017}, and \citet{KimPark-2017} provide prices for the European exchange option under stochastic volatility dynamics. \citet{CheangChiarella-2011} extended Merton's jump-diffusion model to the case of two assets and characterized the price of European exchange options, an analysis which has been refined by \citet{Caldana-2015}. 
\citet{PetroniSabino-2018} consider a market model with correlated jumps in pricing European exchange options.

On the legacy of the fast Fourier transform approaches developed by \citet{DempsterHong-2002} and \citet{HurdZhou-2010}, European spread option prices have also been derived under alternative price processes. \citet{CaneOlivares-2014} used the \citet{HurdZhou-2010} method to price European spread options under a two-dimensional \citet{Bates-1996} model. \citet{AlfeusSchlogl-2018} showed that the \citet{HurdZhou-2010} method is a particular application of the two-dimensional Parseval's Identity. A striking feature of the Fourier transform approach is that closed form expressions or approximations for European spread option prices can be obtained with knowledge of the joint characteristic function of the log-prices of the underlying assets, which is available in most models. 

It was noted however by \citet{CaldanaFusai-2013} that the method of \citet{HurdZhou-2010} is unable to produce prices for European exchange options, so they proposed alternative lower-bound approximations of spread option prices based on the approach of \citet{DempsterHong-2002} as a generalization of the approximations derived by \citet{Bjerskund-2011}. The \citet{CaldanaFusai-2013} result only requires the joint characteristic function of the log-prices of the assets and is exact for European exchange options. 

Analysis of exchange and spread options with non-European payoffs have also gained traction in financial literature. Under the Black-Scholes framework, \citet{Bjerskund-1993} analyzed the American exchange option by approaching it as an optimal stopping problem. \citet{BroadieDetemple-1997} established pricing formulas for multi-asset options, including the exchange and spread options, under the pure-diffusion framework. 
Following their analysis of the European exchange option under jump-diffusion dynamics, \citet{CheangChiarella-2011} obtained a linked system of patrial integro-differential equations characterizing the price of the American exchange option and the associated early exercise boundary. Following the methods of \citet{McKean-1965} and \citet{Jamshidian-1992}, an analysis of American options written on two underlying assets in the pure-diffusion setting was tackled via partial differential equations and Fourier integral transforms by \citet{ChiarellaZiveyi-2014}. \citet{CheangLian-2015} have priced perpetual exchange options under jump-diffusion dynamics, while \citet{PengPeng-2016} analyzed the price of Bermudan-style exchange options under a jump-diffusion model. Typical of American-style derivatives, however, is the unavailability of a closed-form option pricing formula, and so one must resort to approximations or numerical solutions.

In this paper, we focus on the representation of the European and American exchange option prices when the underlying asset prices are characterized by a stochastic volatility jump-diffusion model. After establishing the two-dimensional \citet{Bates-1996} model of the financial market and the necessary change-of-measure mechanisms, we derive the integro-partial differential equation (and the associated boundary conditions) that characterizes exchange option prices under SVJD dynamics. We then employ probabilistic arguments to obtain expressions for the European and American exchange option prices arising from the IPDE.

Our discussion of the European exchange option employs two methods---a probabilistic approach and a Fourier transform approach. First, we demonstrate the use of the \citet{Geman-1995} change-of-num\'eraire technique to obtain an alternative probabilistic representation of European exchange option prices that resembles the original \citet{Margrabe-1978} formula. In this analysis, we also establish some conditions on the volatility processes to ensure that asset and option prices are well-defined after changes in probability measures.
Second, with the imposition of additional assumptions on the correlation structure of the market model, we derive the the joint characteristic function of the log-prices of the stocks (following \citet{ContTankov-2004} and \citet{CaneOlivares-2014}) and use the result of \citet{CaldanaFusai-2013} to obtain the price of the European exchange option. Our analysis shows that the \citet{CaldanaFusai-2013} formulation is compatible with the representation of option prices obtain via the change-of-num\'eraire procedure, linking the probabilities of option exercise (under the secondary probability measures) to Fourier inversion formulas.

To the best of our knowledge, an analysis of American exchange options under both stochastic volatility and jump-diffusion dynamics is yet to be formulated. As such, this paper aims to extend the probabilistic analysis of American exchange options by \citet{CheangChiarella-2011} to the case of stochastic volatility jump-diffusion dynamics. Under the SVJD model, we were able to show that the American exchange option price can also be decomposed into the sum of the price of the European exchange option and an early exercise premium, the latter of which can be decomposed further into diffusive and jump components (a feature that was also shown by \citet{CheangChiarellaZiogas-2013} for the single-asset American option). In this regard, we also derive the coupled system of integral equations that determine both the American exchange option price and the unknown early exercise boundary. 

The rest of the paper is organized as follows. Section \ref{sec-SVJDModel} presents the stochastic volatility jump-diffusion model based on \citet{Bates-1996} and some pertinent results surrounding the stochastic volatility process; Section \ref{sec-RadonNikodymDerivative} presents the Radon-Nikod\'ym derivative that will be used to shift to an equivalent martingale measure as well as to the secondary probability measures required in the change-of-num\'eraire technique; Section \ref{sec-IPDE} contains the derivation of the exchange option pricing integro-partial differential equation; Section \ref{sec-EuropeanExchangeOption} discusses a representation of the European exchange option price based on the change-of-num\'eraire technique (Section \ref{sec-ChangeofNumeraire}) and the Fourier transform method of \citet{CaldanaFusai-2013} (Section \ref{sec-FourierTransform}); Section \ref{sec-AmericanExchangeOption} shows the derivation of the decomposition of the American exchange option price; Section \ref{sec-Conclusion} concludes the paper.

\section{A Stochastic Volatility Jump-Diffusion Model}
\label{sec-SVJDModel}

Let $(\Omega,\calF,\{\calF_t\},\prob)$ be a filtered probability space where $\prob$ is interpreted to be the market probability measure and the filtration $\{\calF_t\}$ is one that is generated by all stochastic process which will be included in the models hereafter.

Let $S_{1,t}$ and $S_{2,t}$ denote the prices of two stocks at time $t$. Denote by $S_{i,t-} = \lim_{u\to t^-}S_{i,u}$ the price of stock $i$ ($i=1,2$) immediately before time $t$; in particular, if a jump in stock price occurs at time $t$, then $S_{i,t-}$ represents the pre-jump price of stock $i$. Assume that stock $i$ pays a (constant) continuously compounded dividend yield $q_i$. We assume that the evolution of the price of stock $i$ is given by a stochastic volatility jump-diffusion (SVJD) model specified as
\begin{align}
\label{eqn-SVJD-ExcOp-Stock}
\dif S_{i,t} & = \mu_i S_{i,t-}\dif t+\sqrt{v_{i,t}}S_{i,t-}\dif W_{i,t}+S_{i,t-}\int_\mathbb{R}(e^{y_i}-1)\left(p(\dif y_i,\dif t)-\lambda_i m_\prob(\dif y_i)\dif t\right)\\
\label{eqn-SVJD-ExcOp-Vol}
\dif v_{i,t} & = \xi_i(\eta_i-v_{i,t})\dif t+\sigma_i\sqrt{v_{i,t}}\dif Z_{i,t},
\end{align}
where $\mu_i$ is the instantaneous return on asset $i$ per unit time, $v_{i,t}$ is the instantaneous variance per unit time, $\xi_i$ is the rate of mean reversion of $v_{i,t}$, $\eta_i$ is the long-run mean for $v_{i,t}$, $\sigma_i$ is the instantaneous volatility of $v_{i,t}$, and $\{W_{i,t}\}$ and $\{Z_{i,t}\}$ are standard Wiener processes under $\prob$. It is assumed that $\mu_i$, $\lambda_i$, $\xi_i$, $\eta_i$, and $\sigma_i$ are positive constants. This assumption is particularly important for the variance process as it ensures that $v_{i,t}$ reverts to a positive level \citep{AndersenPiterbarg-2007}.

Furthermore, the counting measure $p(\dif y_i,\dif t)$ is associated to a marked Poisson process $(Y_{i,n}, N_{i,t})$, where the marks $Y_{i,1},Y_{i,2},\dots$ are i.i.d. random variables with a non-atomic $\prob$-density $m_\prob(\dif y_i)$ and $\{N_{i,t}\}$ is a Poisson process with intensity $\lambda_i$ under $\prob$. In the language of \citet{Runggaldier-2003}, $p(\dif y_i,\dif t)$ has $\prob$-local characteristics $(\lambda_i, m_\prob(\dif y_i))$. It is assumed that the marks $Y_{i,n}$ and the Poisson process $\{N_{i,t}\}$ are independent of each other and independent of the Wiener processes defined above. We further assume that the marks and Poisson processes among the two assets are independent of each other.

This model is a two-asset extension of the stock price dynamics in \citet{CheangChiarellaZiogas-2013}, which combines the jump-diffusion model of \citet{Merton-1976} and the square-root volatility process of \citet{Heston-1993} (as what was done by \citet{Bates-1996}). Likewise, it is a stochastic volatility version of the model introduced in \citet{CheangChiarella-2011} and \citet{CheangLian-2015} for exchange options.

Dependencies in the Wiener components in the stock price equations and volatility equations are assumed to be the following:
\begin{align}
\dif W_{1,t}\dif W_{2,t} & = \rho_w\dif t\\
\dif W_{i,t}\dif Z_{i,t} & = \rho_{wz_i}\dif t, \qquad i=1,2\\
\dif Z_{1,t}\dif Z_{2,t} & = \rho_z\dif t.
\end{align}
Furthermore, we assume that $\dif W_{1,t}\dif Z_{2,t} = \dif W_{1,t}\dif Z_{2,t} = 0$, as correlations between Wiener components in the stock price and volatility processes of stocks 1 and 2 have already been introduced. In tabular form, the correlation structure of the model is summarized as follows:
\begin{center}
\begin{tabular}{c|cccc}
 & $W_{1,t}$ & $W_{2,t}$ & $Z_{1,t}$ & $Z_{2,t}$ \\ \hline
$W_{1,t}$ & 1 & $\rho_w$ & $\rho_{wz_1}$ & 0 \\
$W_{2,t}$ & $\rho_w$ & 1 & 0 & $\rho_{wz_2}$ \\
$Z_{1,t}$ & $\rho_{wz_1}$ & 0 & 1 & $\rho_z$ \\
$Z_{2,t}$ & 0 & $\rho_{wz_2}$ & $\rho_z$ & 1
\end{tabular}
\end{center}
Let $\bm{\Sigma}$ denote the correlation matrix of the Wiener processes as given above. We assume that $\rho_w$ and $\rho_z$ are not equal to $\pm 1$.

\begin{rem}
In the succeeding analysis, the correlation coefficients $\rho_w$ and $\rho_z$ may be set to zero as simplifying assumptions.
\end{rem}

We make the following assumptions on the correlation coefficients and the coefficients of the volatility processes.

\begin{asp}
\label{asp-ParameterAssumptions}
Assume that $Z_{1,t}$ and $Z_{2,t}$ are uncorrelated (i.e. $\rho_z=0$). Furthermore, assume that the coefficients $\xi_i$, $\eta_i$, and $\sigma_i$ are positive and satisfy 
\begin{equation}
2\xi_i\eta_i\geq\sigma_i^2, \qquad i=1,2.
\end{equation}
Lastly, assume that
\begin{equation}
-1 < \rho_{wz_i} < \min\left\{\frac{\xi_i}{\sigma_i},1\right\}, \qquad i=1,2.
\end{equation}
\end{asp}

The condition $2\xi_i\eta_i\geq\sigma_i^2$ is required to ensure that the volatility processes do not hit zero or explode in finite time under the objective market measure $\prob$ \citep{AndersenPiterbarg-2007, Bertini-2008, CheangChiarellaZiogas-2011}. The condition on $\rho_{wz_i}$ ensures that the volatility processes do not do the same under the equivalent martingale measure and the equivalent measures considered in the change-of-num\'eraire procedure. This will be justified in the succeeding sections.

Denote by $\kappa_i$ the expected jump-size increment under $\prob$, which is given by
\begin{equation}
\kappa_i = \int_\mathbb{R}(e^{y_i}-1)m_\prob(\dif y_i).
\end{equation}
Then equation \eqref{eqn-SVJD-ExcOp-Stock} may be rewritten as
\begin{equation}
\label{eqn-SVJD-ExcOp-Stock2}
\dif S_{i,t} = (\mu_i-\lambda_i\kappa_i)S_{i,t-}\dif t+\sqrt{v_{i,t}}S_{i,t-}\dif W_{i,t}+S_{i,t-}\int_\mathbb{R}(e^{y_i}-1)p(\dif y_i,\dif t).
\end{equation}
We assume that the moment generating function of jump sizes, denoted by $M_{\prob,Y_i}(u) = \E_\prob[e^{uY_i}]$, exists but the distribution is not specified.

By It\^{o}'s Lemma for jump-diffusion processes \citep[see][]{Runggaldier-2003}, the dynamics of the log-price $X_{i,t} = \ln S_{i,t}$, provided $S_{i,t}$ satisfies equation \eqref{eqn-SVJD-ExcOp-Stock2}, is given by
\begin{equation}
\dif X_{i,t} = \left(\mu_i-\lambda_i\kappa_i-\frac{1}{2}v_{i,t}\right)\dif t+\sqrt{v_{i,t}}\dif W_{i,t}+\int_{\mathbb{R}}y_i p(\dif y_i,\dif t),
\end{equation}
for $i=1,2$. This implies that equation \eqref{eqn-SVJD-ExcOp-Stock} admits a solution of the form
\begin{equation}
\label{eqn-SVJD-ExcOp-Stock-Soln}
S_{i,t} = S_{i,0}\exp\left\{(\mu_i-\lambda_i\kappa_i)t-\frac{1}{2}\int_0^t v_{i,s}\dif s+\int_0^t \sqrt{v_{i,s}}\dif W_{i,s}+\sum_{n=1}^{N_{i,t}}Y_{i,n}\right\},
\end{equation}
for $i=1,2$. Let $\tilde{S}_{i,t} = e^{-(r-q)t}S_{i,t}$. Then $\{\tilde{S}_{i,t}\}$ gives the discounted yield process for stock $i$ and its $\prob$-dynamics is given by the SDE
\begin{align}
\begin{split}
\dif\tilde{S}_{i,t}
	& = (\mu_i+q_i-r-\lambda_i\kappa_i)\tilde{S}_{i,t-}\dif t+\sqrt{v_{i,t}}\tilde{S}_{i,t-}\dif W_{i,t}\\
	& \qquad + \tilde{S}_{i,t-}\int_\mathbb{R}(e^{y_i}-1)p(\dif y_i,\dif t).
\end{split}
\end{align}
We investigate later the dynamics of the discounted yield process under the equivalent risk-neutral measure in line with the option pricing problem.

With respect to the expressions obtained for the log-price and stock price dynamics, the non-explosion of the volatility process ensures that the integrals $\int_0^t v_{i,s}\dif s$ and $\int_0^t \sqrt{v_{i,s}}\dif W_{i,s}$ are properly defined.\footnote{A discussion of sensible integrands for stochastic integrals can be found in \citet{Shreve-2004} and \citet{Kuo-2006}.} Furthermore, it also guarantees that the process $\{\calM_{i,t}\}$, where
\begin{equation}
\label{eqn-StochExp}
\calM_{i,t} = \exp\left\{-\frac{1}{2}\int_0^t v_{i,s}\dif s+\int_0^t \sqrt{v_{i,s}}\dif W_{i,s}-\lambda_i\kappa_i t+\sum_{n=1}^{N_{i,t}}Y_{i,n}\right\}, \qquad i=1,2,
\end{equation} 
is a $\prob$-martingale (see Appendix \ref{sec-app-MartingaleProperty}).

\section{A Change of Measure Mechanism}
\label{sec-RadonNikodymDerivative}

Let $\vm{B}_t = (W_{1,t},W_{2,t},Z_{1,t},Z_{2,t})^\top$ be a vector of standard $\prob$-Wiener processes with correlation matrix $\bm{\Sigma}$ and let $Q_{1,t} = \sum_{n=1}^{N_{1,t}}Y_{1,n}$ and $Q_{2,t} = \sum_{n=1}^{N_{2,t}} Y_{2,n}$ be compound Poisson processes. In order to achieve a risk-neutral valuation of the price of an exchange option, we must find a suitable Radon-Nikod\'ym derivative that translates the situation from the market measure $\prob$ to a risk-neutral probability measure $\Q$ under which the discounted stock yield processes are martingales.

Let $\bm{\theta}_t = (\psi_{1,t},\psi_{2,t},\zeta_{1,t},\zeta_{2,t})^\top$ be a vector of real-valued adapted processes. The parameters $\psi_{i,t}$ and $\zeta_{i,t}$ are related to the market price of Wiener risk and market price of volatility risk associated to asset $i$, as shall be seen in the succeeding discussion. Assumption \ref{asp-ParameterAssumptions} ensures that these market prices of risks are strictly positive and do not explode in finite time \citep{CheangChiarellaZiogas-2013}.

Following \citet{Runggaldier-2003}, \citet{CheangChiarella-2011}, \citet{CheangChiarellaZiogas-2013}, \citet{CheangTeh-2014}, we state the following Radon-Nikod\'ym derivative to facilitate a change of measure from $\prob$ to $\Q$, inducing a drift in the components of $\vm{B}_t$ under $\Q$, a change in intensity of the Poisson processes, and a change in density of the jump size variables.

\begin{prop}
\label{prop-ChangeofMeasure}
Let $(\Omega,\calF,\{\calF_t\},\prob)$ be a probability space such that $\{\calF_t\}$ is the natural filtration generated by $\vm{B}_t$, $Q_{1,t}$, and $Q_{2,t}$ (as defined above). Let $L_t$ be given by the equation
\begin{align}
\begin{split}
\label{eqn-SVJD-RNDerivative}
L_t	& = \exp\left\{-\int_0^t(\bm{\Sigma}^{-1}\bm{\theta}_s)^\top\dif\vm{B}_s-\frac{1}{2}\int_0^t\bm{\theta}_s^\top\bm{\Sigma}^{-1}\bm{\theta}_s\dif s\right\}\\
		& \qquad \times\exp\left\{\sum_{n=1}^{N_{1,t}}(\gamma_1 Y_{1,n}+\nu_1)-\lambda_1 t\left(e^{\nu_1}\E_\prob(e^{\gamma_1 Y_{1}})-1\right)\right\}\\
		& \qquad \times\exp\left\{\sum_{n=1}^{N_{2,t}}(\gamma_2 Y_{2,n}+\nu_2)-\lambda_2 t\left(e^{\nu_2}\E_\prob(e^{\gamma_2 Y_{2}})-1\right)\right\}
\end{split}
\end{align}
and suppose that $\{L_t\}$ is a strict $\prob$-martingale such that $\E_\prob[L_t]=1$. Then $L_T$ is the Radon-Nikod\'ym derivative of some probability measure $\Q$ equivalent to $\prob$ and the following hold:
\begin{enumerate}
	\item $W_{i,t}$ and $Z_{i,t}$ have drift $-\psi_{i,t}$ and $-\zeta_{i,t}$, respectively for $i=1,2$, under $\Q$;
	\item the compound Poisson process $Q_{i,t} = \sum_{n=1}^{N_{1,t}} Y_{i,n}$ has a new intensity rate
		\begin{equation}
		\tilde{\lambda}_i = \lambda_i e^{\nu_i}\E_\prob[e^{\gamma_i Y_i}], \qquad i=1,2
		\end{equation}
		under $\Q$; and
	\item the moment generating function of jump sizes under $\Q$ is given by
		\begin{equation}
		M_{\Q,Y_i}(u) = \frac{M_{\prob,Y_i}(u+\gamma_i)}{M_{\prob,Y_i}(\gamma_i)}, \qquad i=1,2.
		\end{equation}
\end{enumerate}
\end{prop}

\begin{proof}
The proof is similar to those presented by \citet{Runggaldier-2003} and \citet{CheangTeh-2014}.
\end{proof}

\begin{rem}
In the subsequent analysis, any change of measure will be facilitated by a Radon-Nikod\'ym derivative of the form given in equation \eqref{eqn-SVJD-RNDerivative} and the properties of the new probability measure will be reflected in the choice of parameters $\bm{\theta}_t$, $\gamma_1$, $\gamma_2$, $\nu_1$, and $\nu_2$. 
\end{rem}

Since it is assumed that the jump components in the model and in the Radon-Nikod\'ym deirvative are independent of the Wiener components, the new distributions of these components can be obtained separately from the corresponding components in equation \eqref{eqn-SVJD-RNDerivative} \citep{CheangTeh-2014}. That is, the first factor of equation \eqref{eqn-SVJD-RNDerivative} facilitates the change of measure in the Wiener processes, whereas the last two factors handle the change in distribution of the jump components in the transition to a new probability measure.

We assume that the parameters $\gamma_1$, $\gamma_2$, $\nu_1$, and $\nu_2$ are constant so that the Poisson processes $N_{1,t}$ and $N_{2,t}$ remain homogenous and that the jump sizes remain identically distributed under $\Q$ \citep{CheangChiarellaZiogas-2013}. Also, the independence of $\vm{B}_t$, $Q_{1,t}$, and $Q_{2,t}$ allows the multiplicative nature of the Radon-Nikod\'ym derivative.


\citet{CheangChiarella-2011} proposed a number of ways to select appropriate values of the parameters of the Radon-Nikod\'ym derivative. Among their suggestions is the selection of parameters that induce the minimum entropy martingale measure in \citet{Miyahara-1999}. The observation that there are infinitely many equivalent probability measures also stems from the fact that the market under the SVJD model is incomplete in the \citet{HarrisonPliska-1981} sense. On top of the Wiener components of the stock price dynamics, the stochastic volatility components add two additional sources of randomness from their own Wiener components and the jump components also induce additional randomness \citep{CheangChiarellaZiogas-2013}.

We now investigate the $\Q$-dynamics of the discounted yield processes. From Proposition \ref{prop-ChangeofMeasure}, we can write
\begin{equation}
\dif W_{i,t} = -\psi_{i,t}\dif t+\dif\tilde{W}_{i,t}, \qquad i=1,2,
\end{equation}
where $\{\tilde{W}_{i,t}\}$ is a standard $\Q$-Wiener process. The conjecture also implies that the $\Q$-local characteristics of the counting measure $p(\dif y_i,\dif t)$ are given by $(\tilde{\lambda}_i, m_\Q(\dif y_i))$, where
\begin{equation}
m_\Q(\dif y_i) = \frac{e^{\gamma_i y_i}}{\E_\prob[e^{\gamma_i Y_i}]}m_\prob(\dif y_i),
\end{equation}
for $i=1,2$. As in \citet{Runggaldier-2003}, define the $\Q$-compensated counting measure $q(\dif y_i,\dif t)$ as
\begin{equation}
q(\dif y_i,\dif t) = p(\dif y_i,\dif t)-\tilde{\lambda}_i m_\Q(\dif y_i)\dif t.
\end{equation}
As was shown in the prior section, the $\prob$-dynamics of the discounted yield process $\{\tilde{S}_{i,t}\}$ is given by the SDE $$\dif\tilde{S}_{i,t} = \tilde{S}_{i,t-}\left\{(\mu_i+q_i-r-\lambda_i\kappa_i)\dif t+\sqrt{v}_{i,t}\dif W_{i,t}+\int_\mathbb{R}(e^{y_i}-1)p(\dif y_i,\dif t)\right\}.$$ Substituting the expressions for $\dif W_{i,t}$ and $q(\dif y_i,\dif t)$ above, we obtain
\begin{align*}
\dif\tilde{S}_{i,t}
	& = \tilde{S}_{i,t-}(\mu_i+q_i-r-\lambda_i\kappa_i)\dif t+\sqrt{v_{i,t}}\tilde{S}_{i,t-}(-\psi_{i,t}\dif t+\dif\tilde{W}_{i,t})\\
	& \qquad +\tilde{S}_{i,t-}\int_\mathbb{R}(e^{y_i}-1)\left[q(\dif y_i,\dif t)+\tilde{\lambda}_i m_\Q(\dif y_i)\dif t\right]
\end{align*}
If we let
\begin{equation}
\tilde{\kappa}_i = \int_\mathbb{R}(e^{y_i}-1)m_\Q(\dif y_i) = \E_\Q\left[e^{Y_i}-1\right]
\end{equation}
be the mean relative jump size increment under $\Q$, then $\dif\tilde{S}_{i,t}$ can be simplified to
\begin{align}
\begin{split}
\dif\tilde{S}_{i,t}
	& = \tilde{S}_{i,t-}\left(\mu_i+q_i-r-\lambda_i\kappa_i+\tilde{\lambda}_i\tilde{\kappa}_i-\sqrt{v_{i,t}}\psi_{i,t}\right)\dif t+\sqrt{v_{i,t}}\tilde{S}_{i,t-}\dif\tilde{W}_{i,t}\\
	& \qquad +\tilde{S}_{i,t-}\int_\mathbb{R}(e^{y_i}-1)q(\dif y_i,\dif t).
\end{split}
\end{align}
From here, we can choose the market price of $W_{i,t}$ risk as
\begin{equation}
\psi_{i,t} = \frac{\mu_i+q_i-r-\lambda_i\kappa_i+\tilde{\lambda}_i\tilde{\kappa}_i}{\sqrt{v_{i,t}}},
\end{equation}
which is the risk premium of the stock $\mu_i+q_i-r$ less the jump risk $\lambda_i\kappa_i-\tilde{\lambda}_i\tilde{\kappa}_i$ per unit volatility $\sqrt{v_{i,t}}$. It follows that
\begin{equation}
\dif\tilde{S}_{i,t} = \tilde{S}_{i,t-}\left\{\sqrt{v_{i,t}}\dif\tilde{W}_{i,t}+\int_\mathbb{R}(e^{y_i}-1)q(\dif y_i,\dif t)\right\}.
\end{equation}
and so with this selection for $\psi_{i,t}$, it is clear that the discounted yield process $\{e^{-(r-q_i)t}S_{i,t}\}$ is a martingale under $\Q$.

From the above equation, we can recover the dynamics of the stock price $S_{i,t}$ under $\Q$. Noting that $S_{i,t} = \tilde{S}_{i,t}e^{(r-q_i)t}$, stochastic integration by parts yields
\begin{align*}
\dif S_{i,t}
	& = \tilde{S}_{i,t-}(r-q_i)e^{(r-q_i)t}\dif t+e^{(r-q_i)t}\dif\tilde{S}_{i,t}\\
	& = S_{i,t-}(r-q_i)\dif t+e^{(r-q_i)t}\tilde{S}_{i,t-}\left\{\sqrt{v_{i,t}}\dif\tilde{W}_{i,t}+\int_\mathbb{R}(e^{y_i}-1)q(\dif y_i,\dif t)\right\}\\
	& = S_{i,t-}\left\{(r-q_i)\dif t+\sqrt{v_{i,t}}\dif\tilde{W}_{i,t}+\int_\mathbb{R}(e^{y_i}-1)q(\dif y_i,\dif t)\right\}.
\end{align*}
Equivalently, we can write
\begin{equation}
\label{eqn-QDynamics-Stock}
\dif S_{i,t} = S_{i,t-}\left\{(r-q_i-\tilde{\lambda}_i\tilde{\kappa}_i)\dif t+\sqrt{v_{i,t}}\dif\tilde{W}_{i,t}+\int_\mathbb{R}(e^{y_i}-1)p(\dif y_i,\dif t)\right\}.
\end{equation}
From here, the $\Q$-dynamics of the log-price $X_{i,t}=\ln S_{i,t}$ is given by
\begin{align}
\begin{split}
\dif X_{i,t}
	& = \left(r-q_i-\tilde{\lambda}\tilde{\kappa}_i-\frac{1}{2}v_{i,t}\right)\dif t+\sqrt{v_{i,t}}\dif\tilde{W}_{i,t}+\int_{\mathbb{R}}y_i p(\dif y_i,\dif t). 
\end{split}
\end{align}
From here, it can be seen that the solution $S_{i,t}$, for $0<t\leq T$, to equation \eqref{eqn-QDynamics-Stock} is given by
\begin{equation}
\label{eqn-QSolution-Stock}
S_{i,t} = S_{i,0}\exp\left\{(r-q_i-\tilde{\lambda}_i\tilde{\kappa}_i)t-\frac{1}{2}\int_0^t v_{i,s}\dif s+\int_0^t\sqrt{v_{i,s}}\dif\tilde{W}_{i,s}+\sum_{n=1}^{N_{i,t}}Y_{i,n}\right\}.
\end{equation}

Let $\tilde{Z}_{1,t}$ and $\tilde{Z}_{2,t}$ be standard Wiener processes under $\Q$. Then by Proposition \ref{prop-ChangeofMeasure}, we can write
\begin{equation}
\dif Z_{i,t} = -\zeta_{i,t}\dif t+\dif\tilde{Z}_{i,t}, \qquad i=1,2.
\end{equation}
Thus, the $\Q$-dynamics of the volatility processes are given by
\begin{align*}
\dif v_{i,t}
	& = \xi_i(\eta_i-v_{i,t})\dif t+\sigma_i\sqrt{v_{i,t}}(-\zeta_{i,t}\dif t+\dif\tilde{Z}_{i,t})\\
	& = \left[\xi_i(\eta_i-v_{i,t})-\zeta_{i,t}\sigma_i\sqrt{v_{i,t}}\right]\dif t+\sigma_i\sqrt{v_{i,t}}\dif\tilde{Z}_{i,t}.
\end{align*}
The quantity $\zeta_{i,t}\sigma_i\sqrt{v_{i,t}}$ is interpreted as the market price of volatility risk and is assumed to be independent of the asset price and proportional to current volatility $v_{i,t}$ \citep{Heston-1993}. That is, for some constant $\Lambda_i$, we can write the market price of risk as
\begin{equation}
\zeta_{i,t}\sigma_i\sqrt{v_{i,t}} = \Lambda_i v_{i,t}.
\end{equation}
The constant of proportionality $\Lambda_i$ must be nonnegative to keep consistent with the financial argument that investors demand a positive premium for volatility risk \citep{CheangChiarellaZiogas-2013}. Thus, we can re-express the $\Q$-dynamics of $v_{i,t}$ as
\begin{align}
\begin{split}
\label{eqn-QDynamics-Volatility}
\dif v_{i,t}
	& = \left[\xi_i \eta_i-(\xi_i+\Lambda_i)v_{i,t}\right]\dif t+\sigma_i\sqrt{v_{i,t}}\dif\tilde{Z}_{i,t}\\
	& = (\xi_i+\Lambda_i)\left[\frac{\xi_i\eta_i}{\xi_i+\Lambda_i}-v_{i,t}\right]\dif t+\sigma_i\sqrt{v_{i,t}}\dif\tilde{Z}_{i,t}.
\end{split}
\end{align}

\begin{rem}
The form of $\psi_{i,t}$ and $\zeta_{i,t}$ requires that $v_{i,t}$ is nonzero and finite. These are guaranteed by Assumption \ref{asp-ParameterAssumptions}.
\end{rem}

Assumption \ref{asp-ParameterAssumptions} ensures that, under $\Q$, the volatility processes neither hit zero nor explode. Indeed, if $\xi_i' = \xi_i+\Lambda_i$ and $\eta_i' = \xi_i\eta_i/(\xi_i+\Lambda_i)$, we find that $$2\xi_i'\eta_i' = 2(\xi_i+\Lambda_i)\cdot\frac{\xi_i\eta_i}{\xi_i+\Lambda_i} = 2\xi_i\eta_i\geq \sigma_i^2,$$ where the last inequality is due to Assumption \ref{asp-ParameterAssumptions}.

\section{An Integro-Partial Differential Equation for Exchange Option Prices}
\label{sec-IPDE}


Consider a European exchange option based on two assets with prices $S_{1,t}$ and $S_{2,t}$ such that the final payoff is $(S_{1,T}-S_{2,T})^+$. Due to the Markov property of the vector process $(S_{1,t},S_{2,t},v_{1,t},v_{2,t})^\top$ and the final payoff not being dependent on the entire history of stock prices, the time $t$ price of the European exchange option, denoted by $C_t^E$, is a function of only $t$, $S_{1,t}$, $S_{2,t}$, $v_{1,t}$, and $v_{2,t}$ \citep{ContTankov-2004, CheangChiarellaZiogas-2013}. We can thus write $C_t^E(S_{1,t},S_{2,t},v_{1,t},v_{2,t})$ to denote the time $t$ price of the European exchange option. In the same vein, denote by $C_t^A(S_{1,t},S_{2,t},v_{1,t},v_{2,t})$ the price at time $t$ of the American exchange option. If the European and American exchange options both expire at time $T$, then the terminal payoff is given by
\begin{equation}
C_T^E(S_{1,T},S_{2,T},v_{1,T},v_{2,T}) = C_T^A(S_{1,T},S_{2,T},v_{1,T},v_{2,T}) = (S_{1,T}-S_{2,T})^+.
\end{equation}

Let $\Q$ the the risk-neutral measure determined by the Radon-Nikod\'ym derivative in Proposition \ref{prop-ChangeofMeasure}. Then, the risk-neutral price of the European exchange option is given by
\begin{align}
\begin{split}
\label{eqn-EuropeanExchangeOptionPrice}
C_t^E(S_{1,t},S_{2,t},v_{1,t},v_{2,t})
	& = e^{-r(T-t)}\E_\Q\left[\left.(S_{1,T}-S_{2,T})^+\right|\calF_t\right]\\
	& = e^{-r(T-t)}\E_\Q\left[\left.(S_{1,T}-S_{2,T})^+\right|S_{1,t},S_{2,t},v_{1,t},v_{2,t}\right].
\end{split}
\end{align}
Furthermore, let $\mathscr{T}_T$ denote the collection of all stopping times $\tau$ in the interval $[0,T]$ with respect to the filtration $\{\calF_t\}$. Then the price of the American exchange option is given by \citep{Bjerskund-1993}
\begin{equation}
\label{eqn-AmericanExchangeOptionPrice}
C_t^A(S_{1,t},S_{2,t},v_{1,t},v_{2,t}) = \sup_{\tau\in\mathscr{T}_T}\E_\Q\left[\left.e^{-r(\tau-t)}(S_{1,\tau}-S_{2,\tau})^+\right|S_{1,t},S_{2,t},v_{1,t},v_{2,t}\right].
\end{equation}

In order to apply It\^{o}'s formula for jump-diffusion processes, we require the following assumption on the European and American option price formulas.

\begin{asp}
\label{asp-OptionPriceDifferentiability}
For $t>0$, the functions $$C_t^E(S_{1,t},S_{2,t},v_{1,t},v_{2,t}), \qquad C_t^A(S_{1,t},S_{2,t},v_{1,t},v_{2,t})$$ are at least twice-differentiable in the stock price and volatility variables with continuous second-order partial derivatives. Assume also that these functions have a continuous first-order partial derivative with respect to $t$.
\end{asp}

For now, we do not impose the assumption $\rho_z=0$, as this is not needed in the derivation of the IPDE. We do, however, require that the other conditions in Assumption \ref{asp-ParameterAssumptions} hold.

Given the $\Q$-dynamics of asset prices and volatility processes in equations \eqref{eqn-QDynamics-Stock} and \eqref{eqn-QDynamics-Volatility}, respectively, we can now solve for the stochastic differential equation for exchange option prices $C_t(S_{1,t},S_{2,t},v_{1,t},v_{2,t})$ under $\Q$. Let $$C_{t-} = C_t(S_{1,t-},S_{2,t-},v_{1,t},v_{2,t})$$ denote the pre-jump price of the exchange option in the event that at time $t$ there is a jump in either $S_{1,t}$ or $S_{2,t}$. From the dynamics of the exchange option price, we can then derive the corresponding pricing IPDE, as shown in the next proposition. Note that different boundary conditions will be used to characterize the American exchange option price from the derived IPDE.

\begin{prop}
\label{prop-SVJD-PricingIPDE}
Given the that asset prices and volatilities have $\Q$-dynamics given by equations \eqref{eqn-QDynamics-Stock} and \eqref{eqn-QDynamics-Volatility}, the exchange option price $C_t(S_{1,t},S_{2,t},v_{1,t},v_{2,t})$ satisfies the IPDE
\begin{align}
\begin{split}
\label{eqn-SVJD-PricingIPDE}
rC_{t-}
	& = \calL[C_{t-}]+ \tilde{\lambda}_1\E_\Q^{Y_1}\left[C_t\left(S_{1,t-}e^{Y_1},S_{2,t-},v_{1,t},v_{2,t}\right)-C_{t-}\right]\\
	& \qquad + \tilde{\lambda}_2\E_\Q^{Y_2}\left[C_t\left(S_{1,t-}, S_{2,t-}e^{Y_2},v_{1,t},v_{2,t}\right)-C_{t-}\right],
\end{split}
\end{align}
where the differential operator $\calL$ is defined by
\begin{align}
\begin{split}
\label{eqn-SVJD-IPDEOperator}
\calL[f]
	& = \pder[f]{t}+\sum_{i=1}^2(r-q_i-\tilde{\lambda}_i\tilde{\kappa}_i)S_{i,t-}\pder[f]{s_i}+\sum_{i=1}^2[\xi_i\eta_i-(\xi_i+\Lambda_i)v_{i,t}]\pder[f]{v_i}\\
	& \qquad + \frac{1}{2}\sum_{i=1}^2 v_{i,t}S_{i,t-}^2\pder[^2 f]{s_i^2}+\frac{1}{2}\sum_{i=1}^2 \sigma_i^2 v_{i,t}\pder[^2 f]{v_i^2}+\rho_w\sqrt{v_{1,t} v_{2,t}}S_{1,t-}S_{2,t-}\pder[^2 f]{s_1\partial s_2}\\
	& \qquad + \rho_{wz_1}\sigma_1 v_{1,t}S_{1,t-}\pder[^2 f]{s_1\partial v_1} + \rho_{wz_2}\sigma_2 v_{2,t}S_{2,t-}\pder[^2 f]{s_2\partial v_2}+\rho_Z\sigma_1\sigma_2\sqrt{v_{1,t}v_{2,t}}\pder[^2 f]{v_1\partial v_2},
\end{split}
\end{align}
\end{prop}

Note that the above proposition lacks terminal and boundary conditions to specify the solution of the IPDE. A remark on these conditions will be provided after the proof of the proposition.

\begin{proof}
Using It\^{o}'s formula for jump-diffusion processes \citep[see][]{Runggaldier-2003, Shreve-2004}, we find that $C_t$ satisfies the stochastic differential equation
\begin{align*}
\dif C_t
	& = \left\{\pder[C_{t-}]{t}+\sum_{i=1}^2(r-q_i-\tilde{\lambda}_i\tilde{\kappa}_i)S_{i,t-}\pder[C_{t-}]{s_i}+\sum_{i=1}^2[\xi_i\eta_i-(\xi_i+\Lambda_i)v_{i,t}]\pder[C_{t-}]{v_i}\right.\\
	& \qquad + \frac{1}{2}\sum_{i=1}^2 v_{i,t}S_{i,t-}^2\pder[^2 C_{t-}]{s_i^2}+\frac{1}{2}\sum_{i=1}^2 \sigma_i^2 v_{i,t}\pder[^2 C_{t-}]{v_i^2}+\rho_w\sqrt{v_{1,t} v_{2,t}}S_{1,t-}S_{2,t-}\pder[^2 C_{t-}]{s_1\partial s_2}\\
	& \qquad + \left.\rho_{wz_1}\sigma_1 v_{1,t}S_{1,t-}\pder[^2 C_{t-}]{s_1\partial v_1} + \rho_{wz_2}\sigma_2 v_{2,t}S_{2,t-}\pder[^2 C_{t-}]{s_2\partial v_2}+\rho_Z\sigma_1\sigma_2\sqrt{v_{1,t}v_{2,t}}\pder[^2 C_{t-}]{v_1\partial v_2}\right\}\dif t\\
	& \qquad + \sum_{i=1}^2 \sqrt{v_{i,t}}S_{i,t-}\pder[C_{t-}]{s_i}\dif \tilde{W}_{i,t}+\sum_{i=1}^2\sigma_i\sqrt{v_{i,t}}\pder[C_{t-}]{v_i}\dif \tilde{Z}_{i,t}\\
	& \qquad + \int_\mathbb{R}\left[C_t(S_{1,t-}e^{y_1},S_{2,t-},v_{1,t},v_{2,t})-C_{t-}\right]p(\dif y_1,\dif t)\\
	& \qquad + \int_\mathbb{R}\left[C_t(S_{1,t-},S_{2,t-}e^{y_2},v_{1,t},v_{2,t})-C_{t-}\right]p(\dif y_2,\dif t).
\end{align*}
The counting measure $p(\dif y_i,\dif t)$ can be replaced by the $\Q$-compensated counting measure $q(\dif y_i,\dif t)$, giving us
\begin{align*}
& \int_\mathbb{R}\left[C_t(S_{1,t-}e^{y_1},S_{2,t-},v_{1,t},v_{2,t})-C_{t-}\right]p(\dif y_1,\dif t)\\
& \qquad = \int_\mathbb{R}\left[C_t(S_{1,t-}e^{y_1},S_{2,t-},v_{1,t},v_{2,t})-C_{t-}\right]\left[q(\dif y_1,\dif t)+\tilde{\lambda}_1 m_\Q(\dif y_1)\dif t\right]\\
& \qquad = \tilde{\lambda}_1\E_\Q^{Y_1}\left[C_t\left(S_{1,t-}e^{Y_1},S_{2,t-},v_{1,t},v_{2,t}\right)-C_{t-}\right]\dif t\\
& \qquad \qquad + \int_\mathbb{R}\left[C_t(S_{1,t-}e^{y_1},S_{2,t-},v_{1,t},v_{2,t})-C_{t-}\right]q(\dif y_1,\dif t),
\end{align*}
where $$\E_\Q^{Y_1}\left[C_t(S_{1,t-}e^{Y_1},S_{2,t-},v_{1,t},v_{2,t})-C_{t-}\right] = \int_\mathbb{R}\left[C_t(S_{1,t-}e^{y_1},S_{2,t-},v_{1,t},v_{2,t})-C_{t-}\right]m_\Q(\dif y_1)$$ represents the expected change in the price of the exchange option due to jumps in the price of stock 1. Likewise, we can write
\begin{align*}
& \int_\mathbb{R}\left[C_t(S_{1,t-},S_{2,t-}e^{y_2},v_{1,t},v_{2,t})-C_{t-}\right]p(\dif y_2,\dif t)\\
& \qquad = \tilde{\lambda}_2\E_\Q^{Y_2}\left[C_t\left(S_{1,t-}, S_{2,t-}e^{Y_2},v_{1,t},v_{2,t}\right)-C_{t-}\right]\dif t\\
& \qquad \qquad + \int_\mathbb{R}\left[C_t(S_{1,t-},S_{2,t-}e^{y_2},v_{1,t},v_{2,t})-C_{t-}\right]q(\dif y_2,\dif t)
\end{align*}
We can thus rewrite $\dif C_t$ as
\begin{align*}
\dif C_t
	& = \left\{\pder[C_{t-}]{t}+\sum_{i=1}^2(r-q_i-\tilde{\lambda}_i\tilde{\kappa}_i)S_{i,t-}\pder[C_{t-}]{s_i}+\sum_{i=1}^2[\xi_i\eta_i-(\xi_i+\Lambda_i)v_{i,t}]\pder[C_{t-}]{v_i}\right.\\
	& \qquad + \frac{1}{2}\sum_{i=1}^2 v_{i,t}S_{i,t-}^2\pder[^2 C_{t-}]{s_i^2}+\frac{1}{2}\sum_{i=1}^2 \sigma_i^2 v_{i,t}\pder[^2 C_{t-}]{v_i^2}+\rho_w\sqrt{v_{1,t} v_{2,t}}S_{1,t-}S_{2,t-}\pder[^2 C_{t-}]{s_1\partial s_2}\\
	& \qquad + \rho_{wz_1}\sigma_1 v_{1,t}S_{1,t-}\pder[^2 C_{t-}]{s_1\partial v_1} + \rho_{wz_2}\sigma_2 v_{2,t}S_{2,t-}\pder[^2 C_{t-}]{s_2\partial v_2}+\rho_Z\sigma_1\sigma_2\sqrt{v_{1,t}v_{2,t}}\pder[^2 C_{t-}]{v_1\partial v_2}\\
	& \qquad + \tilde{\lambda}_1\E_\Q^{Y_1}\left[C_t\left(S_{1,t-}e^{Y_1},S_{2,t-},v_{1,t},v_{2,t}\right)-C_{t-}\right]\\
	& \qquad + \left.\tilde{\lambda}_2\E_\Q^{Y_2}\left[C_t\left(S_{1,t-}, S_{2,t-}e^{Y_2},v_{1,t},v_{2,t}\right)-C_{t-}\right] \right\}\dif t\\
	& \qquad + \sum_{i=1}^2 \sqrt{v_{i,t}}S_{i,t-}\pder[C_{t-}]{s_i}\dif \tilde{W}_{i,t}+\sum_{i=1}^2\sigma_i\sqrt{v_{i,t}}\pder[C_{t-}]{v_i}\dif \tilde{Z}_{i,t}\\
	& \qquad + \int_\mathbb{R}\left[C_t(S_{1,t-}e^{y_1},S_{2,t-},v_{1,t},v_{2,t})-C_{t-}\right]q(\dif y_1,\dif t)\\
	& \qquad + \int_\mathbb{R}\left[C_t(S_{1,t-},S_{2,t-}e^{y_2},v_{1,t},v_{2,t})-C_{t-}\right]q(\dif y_2,\dif t).
\end{align*}

If $\tilde{C}_t = e^{-rt}C_t$ represents the discounted exchange option price, then we find that $\tilde{C}_t$ satisfies the SDE
\begin{align}
\begin{split}
\label{eqn-QDynamics-ExcOptionPrice}
\dif\tilde{C}_t
	& = -rC_{t-}e^{-rt}\dif t+e^{-rt}\dif C_t\\
	& = e^{-rt}\left\{\pder[C_{t-}]{t}+\sum_{i=1}^2(r-q_i-\tilde{\lambda}_i\tilde{\kappa}_i)S_{i,t-}\pder[C_{t-}]{s_i}+\sum_{i=1}^2[\xi_i\eta_i-(\xi_i+\Lambda_i)v_{i,t}]\pder[C_{t-}]{v_i}\right.\\
	& \qquad + \frac{1}{2}\sum_{i=1}^2 v_{i,t}S_{i,t-}^2\pder[^2 C_{t-}]{s_i^2}+\frac{1}{2}\sum_{i=1}^2 \sigma_i^2 v_{i,t}\pder[^2 C_{t-}]{v_i^2}+\rho_w\sqrt{v_{1,t} v_{2,t}}S_{1,t-}S_{2,t-}\pder[^2 C_{t-}]{s_1\partial s_2}\\
	& \qquad + \rho_{wz_1}\sigma_1 v_{1,t}S_{1,t-}\pder[^2 C_{t-}]{s_1\partial v_1} + \rho_{wz_2}\sigma_2 v_{2,t}S_{2,t-}\pder[^2 C_{t-}]{s_2\partial v_2}+\rho_Z\sigma_1\sigma_2\sqrt{v_{1,t}v_{2,t}}\pder[^2 C_{t-}]{v_1\partial v_2}\\
	& \qquad + \tilde{\lambda}_1\E_\Q^{Y_1}\left[C_t\left(S_{1,t-}e^{Y_1},S_{2,t-},v_{1,t},v_{2,t}\right)-C_{t-}\right]\\
	& \qquad + \left.\tilde{\lambda}_2\E_\Q^{Y_2}\left[C_t\left(S_{1,t-}, S_{2,t-}e^{Y_2},v_{1,t},v_{2,t}\right)-C_{t-}\right] - rC_{t-}\right\}\dif t\\
	& \qquad + e^{-rt}\sum_{i=1}^2 \sqrt{v_{i,t}}S_{i,t-}\pder[C_{t-}]{s_i}\dif \tilde{W}_{i,t}+e^{-rt}\sum_{i=1}^2\sigma_i\sqrt{v_{i,t}}\pder[C_{t-}]{v_i}\dif \tilde{Z}_{i,t}\\
	& \qquad + e^{-rt}\int_\mathbb{R}\left[C_t(S_{1,t-}e^{y_1},S_{2,t-},v_{1,t},v_{2,t})-C_{t-}\right]q(\dif y_1,\dif t)\\
	& \qquad + e^{-rt}\int_\mathbb{R}\left[C_t(S_{1,t-},S_{2,t-}e^{y_2},v_{1,t},v_{2,t})-C_{t-}\right]q(\dif y_2,\dif t)
\end{split}
\end{align}
The non-explosion of the volatility processes (implied by Assumption \ref{asp-ParameterAssumptions}) and the differentiability of the option price (Assumption \ref{asp-OptionPriceDifferentiability}) ensure that $$\E_\Q\left[\int_0^t \left|e^{-rs}\sqrt{v_{i,s}}S_{i,s-}\pder[C_{s-}]{s_i}\right|^2\dif s\right]<\infty \quad\text{and}\quad \E_\Q\left[\int_0^t\left|e^{-rs}\sigma_i\sqrt{v_{i,s}}\pder[C_{s-}]{v_i}\right|^2\dif s\right]<\infty$$ hold for $i=1,2$, and so processes whose stochastic differentials correspond to the $\dif \tilde{W}_{i,t}$ and $\dif\tilde{Z}_{i,t}$ terms are $\Q$-martingales \citep[see][Theorem 4.6.1]{Kuo-2006}. Furthermore, \citet[Theorem 2.2]{Runggaldier-2003} ensures that the last two terms of the right-hand side of the above equation correspond to $\Q$-martingales, provided integrability conditions hold for the option price increments.

Under $\Q$, we require $\tilde{C}_t$ to have no drift. Setting the coefficient of $\dif t$ in equation \eqref{eqn-QDynamics-ExcOptionPrice} to zero, we find that the exchange option price satisfies the IPDE
\begin{align}
\begin{split}
rC_{t-}
	& = \pder[C_{t-}]{t}+\sum_{i=1}^2(r-q_i-\tilde{\lambda}_i\tilde{\kappa}_i)S_{i,t-}\pder[C_{t-}]{s_i}+\sum_{i=1}^2[\xi_i\eta_i-(\xi_i+\Lambda_i)v_{i,t}]\pder[C_{t-}]{v_i}\\
	& \qquad + \frac{1}{2}\sum_{i=1}^2 v_{i,t}S_{i,t-}^2\pder[^2 C_{t-}]{s_i^2}+\frac{1}{2}\sum_{i=1}^2 \sigma_i^2 v_{i,t}\pder[^2 C_{t-}]{v_i^2}+\rho_w\sqrt{v_{1,t} v_{2,t}}S_{1,t-}S_{2,t-}\pder[^2 C_{t-}]{s_1\partial s_2}\\
	& \qquad + \rho_{wz_1}\sigma_1 v_{1,t}S_{1,t-}\pder[^2 C_{t-}]{s_1\partial v_1} + \rho_{wz_2}\sigma_2 v_{2,t}S_{2,t-}\pder[^2 C_{t-}]{s_2\partial v_2}+\rho_Z\sigma_1\sigma_2\sqrt{v_{1,t}v_{2,t}}\pder[^2 C_{t-}]{v_1\partial v_2}\\
	& \qquad + \tilde{\lambda}_1\E_\Q^{Y_1}\left[C_t\left(S_{1,t-}e^{Y_1},S_{2,t-},v_{1,t},v_{2,t}\right)-C_{t-}\right]\\
	& \qquad + \tilde{\lambda}_2\E_\Q^{Y_2}\left[C_t\left(S_{1,t-}, S_{2,t-}e^{Y_2},v_{1,t},v_{2,t}\right)-C_{t-}\right]
\end{split}
\end{align}
Using the differential operator $\calL$, the preceding IPDE may be written as
\begin{align*}
rC_{t-}
	& = \calL[C_{t-}]+ \tilde{\lambda}_1\E_\Q^{Y_1}\left[C_t\left(S_{1,t-}e^{Y_1},S_{2,t-},v_{1,t},v_{2,t}\right)-C_{t-}\right]\\
	& \qquad + \tilde{\lambda}_2\E_\Q^{Y_2}\left[C_t\left(S_{1,t-}, S_{2,t-}e^{Y_2},v_{1,t},v_{2,t}\right)-C_{t-}\right].
\end{align*}
\end{proof}

The IPDE derived above extends the result obtained by \citet{CheangChiarella-2011} for exchange options under jump-diffusion dynamics to the case of stochastic volatility and jump-diffusion dynamics. It is also an extension of the IPDE derived by \citet{CheangChiarellaZiogas-2013} for the one-asset option under SVJD dynamics to the case of two risky assets.

For the European exchange option, the terminal condition for the IPDE is $$C^E_T = (S_{1,T}-S_{2,T})^+.$$ In the case of the American exchange option, additional conditions, namely the early exercise boundary condition and smooth-pasting conditions, must be specified given the early exercise boundary of the option \citep{Chiarella-2009, CheangChiarella-2011, CheangChiarellaZiogas-2013, ChiarellaKangMeyer-2015}. These will be discussed in Section \ref{sec-AmericanExchangeOption}.

\section{A Representation of the European Exchange Option Price}
\label{sec-EuropeanExchangeOption}

In this section, we now consider analytical representations of the price of the European exchange option. To this end, we employ two methods: the change-of-num\'eraire technique of \citet{Geman-1995}, which was applied by \citet{CheangChiarella-2011} to the exchange option, and the Fourier transform approach by \citet{DempsterHong-2002} and \citet{CaldanaFusai-2013}. For the first approach, we assume that $\rho_z=0$ and that Assumption \ref{asp-ParameterAssumptions} hold. In the latter approach, we require additional restrictions on the correlation structure of the Wiener processes.


\subsection{A Change-of-Num\'eraire Approach}
\label{sec-ChangeofNumeraire}

In this section, we employ the change of num\'eraire technique of \citet{Geman-1995} to evaluate the $\Q$-expectation that gives the price of the European exchange option. This technique has been employed in \citet{CheangChiarella-2011} and \citet{Caldana-2015} to price exchange options under jump-diffusion dynamics, \citet{CheangTeh-2014} to price single-asset options under jump-diffusion dynamics with stochastic interest rate, and \citet{CheangChiarellaZiogas-2011} and \citet{CheangChiarellaZiogas-2013} to price single-asset European options under SVJD dynamics. Here, we derive a representation for European exchange option prices under the two-asset SVJD model.

Without loss of generality, we analyze the European exchange option price at time $t=0$, which is given by $$C_0^E(S_{1,0},S_{2,0},v_{1,0},v_{2,0}) = e^{-rT}\E_\Q\left[(S_{1,T}-S_{2,T})^+\right].$$ Define the event $\calA_0 = \{S_{1,T}>S_{2,T}\}$ (the event that the option is in-the-money), so that we can write
\begin{equation}
\label{eqn-EuExcOp-Price1}
C_0^E	= e^{-rT}\E_\Q[S_{1,T}\vm{1}_{\calA_0}]-e^{-rT}\E_\Q[S_{2,T}\vm{1}_{\calA_0}].
\end{equation}
Using equation \eqref{eqn-QSolution-Stock}, we may substitute expressions for $S_{1,T}$ and $S_{2,T}$, giving us
\begin{equation}
\label{eqn-EuExcOp-Price2}
C_0^E	= S_{1,0}e^{-q_1 T}\E_\Q[U_{1,T}\vm{1}_{\calA_0}]-S_{2,0}e^{-q_2 T}\E_\Q[U_{2,T}\vm{1}_{\calA_0}],
\end{equation}
where
\begin{equation}
\label{eqn-EuExcOp-RNDerivative}
U_{i,T} = \exp\left\{-\frac{1}{2}\int_0^T v_{i,t}\dif t+\int_0^T\sqrt{v_{i,t}}\dif\tilde{W}_{i,t}-\tilde{\lambda}_i\tilde{\kappa}_i T+\sum_{n=1}^{N_{i,T}}Y_{i,n}\right\}.
\end{equation}
Therefore, the expectations may be seen as the probability that the option is in the money at time $T$ under two new probability measures $\hat{\Q}_1$ and $\hat{\Q}_2$ whose relative densities with respect to $\Q$ are $U_{1,T}$ and $U_{2,T}$, respectively. The probability measures $\hat{\Q}_1$ and $\hat{\Q}_2$ are those that result from using $S_{1,t}$ and $S_{2,t}$, respectively, as the num\'eraire.

To show that $U_{1,T}$ and $U_{2,T}$ are sensible Radon-Nikod\'ym derivatives in the sense of Proposition \ref{prop-ChangeofMeasure}, we must show that $U_{1,t}$ and $U_{2,t}$ are $\Q$-martingales and $\E_\Q[U_{i,t}]=1$, $i=1,2$. This implies that we must ensure that the volatility processes do not explode under the new probability measures $\hat{\Q}_1$ and $\hat{\Q}_2$ (which in effect guarantees the existence of an expression for asset prices in the new probability measures).

The subsequent analysis is for $\{U_{1,t}\}$ and is also applicable to $\{U_{2,t}\}$. As shown in Section \ref{sec-RadonNikodymDerivative}, the condition $2\xi_1\eta_1\geq\sigma_1^2$ on the $\prob$-dynamics of $\{v_{1,t}\}$ is sufficient to ensure that the process neither explodes nor makes excursions to the origin under $\Q$. As such, $\int_0^t v_{1,s}\dif s<\infty$ $\Q$-a.s., which satisfies the Novikov condition to ensure that $\exp\{-\frac{1}{2}\int_0^t v_{1,s}\dif s+\int_0^t \sqrt{v_{1,s}}\dif W_{1,s}\}$ is a $\Q$-martingale and $$\E_\Q\left[\exp\left\{-\frac{1}{2}\int_0^t v_{1,s}\dif s+\int_0^t \sqrt{v_{1,s}}\dif W_{1,s}\right\}\right] = 1.$$ We also note that $$\E_\Q\left[\exp\left\{\sum_{i=1}^{N_{1,t}} Y_{1,n}\right\}\right] = \exp\left\{\tilde{\lambda}_1 t(\E_\Q(e^{Y_i})-1)\right\} = e^{\tilde{\lambda}_1\tilde{\kappa}_1 t}.$$ Independence of the Wiener and jump components imply that
\begin{align*}
\E_\Q[U_{1,t}]
	& = \E_\Q\left[\exp\left\{-\frac{1}{2}\int_0^t v_{1,s}\dif s+\int_0^t \sqrt{v_{1,s}}\dif W_{1,s}\right\}\right]\cdot e^{-\tilde{\lambda}_1\tilde{\kappa}_1 t}\cdot\E_\Q\left[\exp\left\{\sum_{i=1}^{N_{1,t}} Y_{1,n}\right\}\right]\\
	& = 1 = U_{1,0}.
\end{align*}
This also shows that $\{U_{1,t}\}$ is a $\Q$-martingale.

By Proposition \ref{prop-ChangeofMeasure}, $U_{1,T}$ defines a Radon-Nikod\'ym derivative that facilitates a change of measure from $\Q$ to some equivalent measure $\hat{\Q}_1$. To determine the drift of the Wiener processes and the distributional properties of the jumps under $\hat{\Q}_1$, we compare equation \eqref{eqn-EuExcOp-RNDerivative} (with $i=1$) with equation \eqref{eqn-SVJD-RNDerivative} to determine the change of measure parameters. 

We first investigate changes in the jump components arising from the shift to $\hat{\Q}_1$. Note that the jump component of $S_{2,t}$ does not appear in $U_{1,T}$, implying that the change of measure has no effect on the jumps in stock 2. From the comparison, we find that $\gamma_1 = 1$ and $\nu_1 = 0$, which implies that the Poisson process $\{N_{i,t}\}$ has $\hat{\Q}_1$-intensity $$\hat{\lambda}^{(1)}_1 = \tilde{\lambda}_1\E_\Q(e^{Y_1}) = \tilde{\lambda}_1(1+\tilde{\kappa}_1)$$ and the new distribution of the jump random variables $Y_{1,n}$ is given by the moment generating function $$M_{\hat{\Q}_1,Y_1}(u) = \frac{M_{\Q,Y_1}(u+1)}{M_{\Q,Y_1}(1)}.$$ This relation between the moment generating functions also implies that the $\hat{\Q}_1$-density of $Y_1$ is given by $$m_{\hat{\Q}_1}(\dif y_1) = \frac{e^{y_1}}{\E_\Q(e^{Y_1})}m_\Q(\dif y_1).$$ This analysis therefore implies that the compensated counting measures under $\hat{\Q}_1$ corresponding to the original counting measure $p(\dif y_i,\dif t)$ are given by
\begin{align}
\begin{split}
\hat{q}^{(1)}(\dif y_1,\dif t) & = p(\dif y_1,\dif t)-\hat{\lambda}^{(1)}_1 m_{\hat{\Q}_1}(\dif y_1)\dif t\\
\hat{q}^{(1)}(\dif y_2,\dif t) & = p(\dif y_2,\dif t)-\hat{\lambda}^{(1)}_2 m_{\hat{\Q}_1}(\dif y_2)\dif t,
\end{split}
\end{align}
where $\hat{\lambda}^{(1)}_1$ and $m_{\hat{\Q}_1}(\dif y_1)$ are given above, and $\hat{\lambda}^{(1)}_2 = \tilde{\lambda}_2$ and $m_{\hat{\Q}_1}(\dif y_2) = m_\Q(\dif y_2)$ as no changes are introduced to the jump components of stock 2.

Now we consider the diffusion components. The parameter $\hat{\bm{\theta}}_t^{(1)}$ for this change of measure (analogous to $\bm{\theta}_1$ in equation \eqref{eqn-SVJD-RNDerivative}) is defined such that $$\left(\bm{\Sigma}^{-1}\hat{\bm{\theta}}_t^{(1)}\right)^\top\dif\tilde{\vm{B}}_t = \sqrt{v_{1,t}}\dif\tilde{W}_{1,t},$$ where $\dif\tilde{\vm{B}}_t = (\dif\tilde{W}_{1,t}, \dif\tilde{W}_{2,t},\dif\tilde{Z}_{1,t},\dif\tilde{Z}_{2,t})^\top$ is the vector of the $\Q$-Wiener increments and $\bm{\Sigma}$ is the original correlation matrix of the Wiener processes. The above equation, in matrix form, can also be written as
$$\left[\begin{array}{cccc}
1	& \rho_w & \rho_{wz_1} & 0\\
\rho_w & 1 & 0 & \rho_{wz_2}\\
\rho_{wz_1} & 0 & 1 & \rho_z\\
0 & \rho_{wz_2} & \rho_z & 1
\end{array}\right]^{-1}\hat{\bm{\theta}}_t^{(1)} = 
\left[\begin{array}{c}
\sqrt{v_{1,t}} \\ 0 \\ 0 \\ 0
\end{array}\right].$$
This implies that $$\hat{\bm{\theta}}_t^{(1)} = \left(\sqrt{v_{1,t}},\rho_w\sqrt{v_{1,t}},\rho_{wz_1}\sqrt{v_{1,t}},0\right)^\top.$$ Thus, if $\hat{W}^{(1)}_{1,t}$, $\hat{W}^{(1)}_{2,t}$, $\hat{Z}^{(1)}_{1,t}$, and $\hat{Z}^{(1)}_{2,t}$ are standard $\hat{\Q}_1$-Wiener processes, then the $\hat{\Q}_1$-dynamics of the $\Q$-Wiener processes are given by
\begin{align}
\begin{split}
\dif\tilde{W}_{1,t}	& = -\sqrt{v_{1,t}}\dif t+\dif\hat{W}^{(1)}_{1,t}\\
\dif\tilde{W}_{2,t}	& = -\rho_w\sqrt{v_{1,t}}\dif t+\dif\hat{W}^{(1)}_{2,t}\\
\dif\tilde{Z}_{1,t}	& = -\rho_{wz_1}\sqrt{v_{1,t}}\dif t+\dif\hat{Z}^{(1)}_{1,t}\\
\dif\tilde{Z}_{2,t}	& = \dif\hat{Z}^{(1)}_{2,t}.
\end{split}
\end{align}

Under $\hat{\Q}_1$, the variance process $v_{1,t}$ satisfies the equation
\begin{equation}
\dif v_{1,t} = (\xi_1+\Lambda_1-\sigma_1\rho_{wz_1})\left[\frac{\xi_1\eta_1}{\xi_1+\Lambda_1-\sigma_1\rho_{wz_1}}-v_{1,t}\right]\dif t+\sigma_1\sqrt{v_{1,t}}\dif\hat{Z}^{(1)}_{1,t},
\end{equation}
which is obtained by substituting $\dif\tilde{Z}_{1,t} = -\rho_{wz_1}\sqrt{v_{1,t}}\dif t+\dif\hat{Z}^{(1)}_{1,t}$ into equation \eqref{eqn-QDynamics-Volatility}. In light of Assumption \ref{asp-ParameterAssumptions}, we require that $\xi_1+\Lambda_1-\sigma_1\rho_{wz_1}>0$. From here and from the fact that $\Lambda_1$ is chosen to be a nonnegative constant, we see that the condition
\begin{equation}
-1<\rho_{wz_1}<\min\left\{\frac{\xi_1}{\sigma_1},1\right\}
\end{equation}
is sufficient to ensure that the coefficients of the volatility process under $\hat{\Q}_1$ are positive. Furthermore, we observe that $$2\left(\xi_1+\Lambda_1-\sigma_1\rho_{wz_1}\right)\cdot\frac{\xi_1\eta_1}{\xi_1+\Lambda_1-\sigma_1\rho_{wz_1}} = 2\xi_1\eta_1 \geq \sigma_1^2,$$ which means that $v_{1,t}$ neither explodes in finite time or makes excursions to 0 under $\hat{\Q}_1$.

Similarly, the $\hat{\Q}_1$ dynamics of $v_{2,t}$ is given by $$\dif v_{2,t} = (\xi_2+\Lambda_2)\left[\frac{\xi_2\eta_2}{\xi_2+\Lambda_2}-v_{2,t}\right]\dif t+\sigma_2\sqrt{v_{2,t}}\dif\hat{Z}^{(1)}_{2,t},$$ which we note to be identical to equation \eqref{eqn-QDynamics-Volatility} except for the change in the Wiener process. At this point, no further parameter assumptions are required aside from those in Assumption \ref{asp-ParameterAssumptions} to ensure that $v_{2,t}$ does not vanish or explode in finite time under $\hat{\Q}_1$.

An analysis similar to that above also shows that $\{U_{2,t}\}$ is a $\Q$-martingale, allowing us to define a new probability measure $\hat{\Q}_2$ equivalent to $\Q$ via the Radon-Nikod\'ym derivative $U_{2,T}$. Under $\hat{\Q}_2$, the local characteristics $(\hat{\lambda}^{(2)}_1,m_{\hat{\Q}_2}(\dif y_1))$ of the counting measure $p(\dif y_1,\dif t)$ are given by $$\hat{\lambda}^{(2)}_1 = \tilde{\lambda}_1, \qquad m_{\hat{\Q}_2}(\dif y_1) = m_\Q(\dif y_1),$$ as $U_{2,T}$ is parameterized such that no changes are induced on the distributional properties of the jump component of stock 1. Analogous to the above analysis, we find that the $\hat{\Q}_2$-local characteristics $(\hat{\lambda}^{(2)}_2,m_{\hat{\Q}_2}(\dif y_2))$ of $p(\dif y_2,\dif t)$ are given by
\begin{align*}
\hat{\lambda}^{(2)}_2 & = \tilde{\lambda}_2(1+\tilde{\kappa}_2)\\
m_{\hat{\Q}_2}(\dif y_2) & = \frac{e^{y_2}}{\E_\Q(e^{Y_2})}m_\Q(\dif y_2).
\end{align*}
A comparison between equation \eqref{eqn-EuExcOp-RNDerivative} (with $i=2$) with equation \eqref{eqn-SVJD-RNDerivative} produces $$\hat{\bm{\theta}}_t^{(2)} = \left(\rho_w\sqrt{v_{2,t}},\sqrt{v_{2,t}},0,\rho_{wz_2}\sqrt{v_{2,t}}\right)^\top,$$ to facilitate the change of drift upon shifting to $\hat{\Q}_2$. Let $\hat{W}^{(2)}_{1,t}$, $\hat{W}^{(2)}_{2,t}$, $\hat{Z}^{(2)}_{1,t}$, and $\hat{Z}^{(2)}_{2,t}$ be standard $\hat{\Q}_2$-Wiener processes; then the $\hat{\Q}_2$-dynamics of the $\Q$-Wiener processes are given by
\begin{align}
\begin{split}
\dif\tilde{W}_{1,t}	& = -\rho_w\sqrt{v_{2,t}}\dif t+\dif\hat{W}^{(2)}_{1,t}\\
\dif\tilde{W}_{2,t}	& = -\sqrt{v_{2,t}}\dif t+\dif\hat{W}^{(2)}_{2,t}\\
\dif\tilde{Z}_{1,t}	& = \dif\hat{Z}^{(2)}_{1,t}\\
\dif\tilde{Z}_{2,t}	& = \rho_{wz_2}\sqrt{v_{2,t}}\dif t+ \dif\hat{Z}^{(2)}_{2,t}.
\end{split}
\end{align}
Consequently, the $\hat{\Q}_2$-dynamics of $v_{2,t}$ is given by the equation
\begin{equation}
\dif v_{2,t} = \left(\xi_2+\Lambda_2-\sigma_2\rho_{wz_2}\right)\left[\frac{\xi_2\eta_2}{\xi_2+\Lambda_2-\sigma_2\rho_{wz_2}}-v_{2,t}\right]\dif t+\sigma_2\sqrt{v_{2,t}}\dif\hat{Z}^{(2)}_{2,t},
\end{equation}
which is guaranteed to neither explode in finite time nor make excursions to 0 under $\hat{\Q}_2$ by enforcing the condition
\begin{equation}
-1<\rho_{wz_2}<\min\left\{\frac{\xi_2}{\sigma_2},1\right\}.
\end{equation}
No additional restrictions need to be made to ensure that $v_{1,t}$, aside from those in Assumption \ref{asp-ParameterAssumptions}, does not explode nor go to zero under $\hat{\Q}_2$.

\begin{rem}
The preceding discussion indeed show that Assumption \ref{asp-ParameterAssumptions}, in the case that $\rho_z=0$, is sufficient to ensure that the volatility processes $\{v_{1,t}\}$ and $\{v_{2,t}\}$ neither explode in finite time nor hit zero under the risk-neutral measure $\Q$ and the probability measures $\hat{\Q}_1$ and $\hat{\Q}_2$ equivalent to $\Q$ determined by the Radon-Nikod\'ym derivatives $U_{1,T}$ and $U_{2,T}$. This extends the discussion of \citet{CheangChiarellaZiogas-2011} to the case of two assets modelled with stochastic volatility and jump-diffusion dynamics.
\end{rem}

In terms of the new probability measures $\hat{\Q}_1$ and $\hat{\Q}_2$, the price of the European exchange option may be written as
\begin{equation}
\label{eqn-EuExcOp-Price-Decomposed}
C_0^E = S_{1,0}e^{-q_1 T}\hat{\Q}_1({\calA_0})-S_{2,0}e^{-q_2 T}\hat{\Q}_2({\calA_0}).
\end{equation}
Using equation \eqref{eqn-QSolution-Stock}, the event $\calA_0$ may be rewritten as
\begin{equation}
\left\{\mathfrak{A}_{0,T} > \ln\left(\frac{S_{2,0}}{S_{1,0}}\right)-\left(q_2-q_1-\tilde{\lambda}_1\tilde{\kappa}_1+\tilde{\lambda}_2\tilde{\kappa}_2\right)T\right\},
\end{equation}
where $\mathfrak{A}_{0,T}$ is the random variable
\begin{equation}
\mathfrak{A}_{0,T} = -\frac{1}{2}\int_0^T(v_{1,t}-v_{2,t})\dif t+\int_0^T\sqrt{v_{1,t}}\dif\tilde{W}_{1,t}-\int_0^T\sqrt{v_{2,t}}\dif\tilde{W}_{2,t}+\sum_{n=1}^{N_{1,T}}Y_{1,n}-\sum_{n=1}^{N_{2,T}}Y_{2,n}.
\end{equation}

\begin{rem}
Similar to the original \citet{Margrabe-1978} formula, our characterization of the European exchange option price under SVJD dynamics is also independent of the risk-free interest rate $r$.
\end{rem}

In general, for any $0\leq t<T$, the time-$t$ price of the European exchange option is given by
\begin{equation}
\label{eqn-EuExcOp-Price-Decomposed-Timet}
C_t^E = S_{1,t}e^{-q_1(T-t)}\hat{\Q}_1(\calA_t)-S_{2,t}e^{-q_2(T-t)}\hat{\Q}_2(\calA_t),
\end{equation}
where $\calA_t$ is the event $$\left\{\mathfrak{A}_{t,T} > \ln\left(\frac{S_{2,t}}{S_{1,t}}\right)-\left(q_2-q_1-\tilde{\lambda}_1\tilde{\kappa}_1+\tilde{\lambda}_2\tilde{\kappa}_2\right)(T-t)\right\}$$ and $\mathfrak{A}_{t,T}$ is the random variable $$\mathfrak{A}_{t,T} = -\frac{1}{2}\int_t^T(v_{1,s}-v_{2,s})\dif s+\int_t^T\sqrt{v_{1,s}}\dif\tilde{W}_{1,s}-\int_t^T\sqrt{v_{2,s}}\dif\tilde{W}_{2,s}+\sum_{n=N_{1,t}}^{N_{1,T}}Y_{1,n}-\sum_{n=N_{2,t}}^{N_{2,T}}Y_{2,n}.$$ Here, the summations count the number of jumps that occur in the period $(t,T]$. By construction, $C_t^E$ as given above is a solution to the IPDE in Proposition \ref{prop-SVJD-PricingIPDE} subject to the terminal condition $C_T^E = (S_{1,T}-S_{2,T})^+$.

In the presence of stochastic volatilities $v_{1,t}$ and $v_{2,t}$, a series expression for the European exchange option price similar to those obtained in \citet{CheangChiarella-2011} and \citet{Caldana-2015} cannot be obtained. In light of our general correlation structure for the Wiener processes, the probabilities above may be computed via simulation methods. Alternatively, a solution via characteristic functions can be made possible by making additional assumptions on the relationship between the stock prices and the volatilties, as will be shown in the next section. In the next section, we will derive a representation of the European exchange option price in terms of Fourier inversion formulas whose form is comparable to equations \eqref{eqn-EuExcOp-Price1} and \eqref{eqn-EuExcOp-Price-Decomposed}.

\subsection{A Fourier Transform Approach}
\label{sec-FourierTransform}

One may derive a Fourier inversion formula for the price of the European exchange option in terms of the joint characteristic function of the log-prices of the stocks. However, an explicit formula for the characteristic function is available if we make additional simplifying assumptions to the correlation structure of the Wiener processes. The derivation of the characteristic function and the application of the \citet{CaldanaFusai-2013} result to derive the exact price of the European exchange option is the main topic of this section. Here, we show that the \citet{CaldanaFusai-2013} result also lends itself to a decomposition similar to equation \eqref{eqn-EuExcOp-Price-Decomposed} in the case of the European exchange option.

Recall that the risk-neutral dynamics of the log-price of the stock is given by
\begin{equation}
\dif X_{i,t} = \left(r-q_i-\tilde{\lambda}_i\tilde{\kappa}_i-\frac{1}{2}v_{i,t}\right)\dif t+\sqrt{v_{i,t}}\dif\tilde{W}_{i,t}+\dif\tilde{Q}_{i,t}, \quad i=1,2
\end{equation}
where $\tilde{Q}_{i,t}$ is a compound Poisson process with intensity $\tilde{\lambda}_i$ under $\Q$ whose jump components $Y_{i,1},Y_{i,2},\dots$ independent and identically distributed with common characteristic function $\phi_{Y_i}(u)$. We also recall that under $\Q$, the volatility processes $v_{i,t}$ satisfy the equation
\begin{equation}
\dif v_{i,t} = \left[\xi_i\eta_i-(\xi_i+\Lambda_i)v_{i,t}\right]\dif t+\sigma_i\sqrt{v_{i,t}}\dif\tilde{Z}_{i,t},
\end{equation}
where the model parameters follow Assumption \ref{asp-ParameterAssumptions}.

Denote by $X_{i,t}^c$ the continuous part of $X_{i,t}$ such that it satisfies the equation
\begin{equation}
\dif X_{i,t}^c = \left(r-q_i-\tilde{\lambda}_i\tilde{\kappa}_i-\frac{1}{2}v_{i,t}\right)\dif t+\sqrt{v_{i,t}}\dif\tilde{W}_{i,t}.
\end{equation}
At this juncture, we follow the steps in \citet{ContTankov-2004} and \citet{CaneOlivares-2014} in deriving the joint characteristic function of $X_{1,t}^c$ and $X_{2,t}^c$ and, eventually, the joint characteristic function of the log-prices.

\begin{lem}
\label{lem-SVJD-LogPriceCF}
Suppose $\rho_w=\rho_z=0$. The joint characteristic function $\phi_{X_t}(u_1,u_2)$ of the log-price vector $X_t = (X_{1,t},X_{2,T})^\top$ is given by
\begin{align}
\begin{split}
\label{eqn-SVJD-LogPriceCF}
\phi_{X_t}(u_1,u_2)
	& = \exp\left\{i(u_1 x_1+u_2 x_2)+C(t;u_1,u_2)+D_1(t;u_1)v_1+D_2(t;u_2)v_2\right\}\\
	& \qquad\times \exp\left\{\tilde{\lambda}_1 t\left(\phi_{Y_1}(u_1)-1\right)+\tilde{\lambda}_2 t\left(\phi_{Y_2}(u_2)-1\right)\right\}
\end{split}
\end{align}
where 
\begin{equation}
D_j(s;u_j) = -\frac{u_j^2+iu_j}{\gamma_j\coth(\gamma_j s/2)+(\xi_j+\lambda_j)-iu_j\rho_{wz_j}\sigma_j}, \qquad j=1,2,
\end{equation}
where $\gamma_j$ is given by $$\gamma_j = \sqrt{\sigma_j^2(u_j^2+iu_j)+(\xi_j+\Lambda_j-iu_j\rho_{wz_j}\sigma_j)^2}, \qquad j=1,2,$$ and
\begin{align}
\begin{split}
C(s;u_1,u_2) 	
			& = \sum_{j=1}^2\left\{iu_j s(r-q_j-\tilde{\lambda}_j\tilde{\kappa}_j)+\frac{\xi_j\eta_j s (\xi_j+\lambda_j-iu_j\rho_{wz_j}\sigma_j)}{\sigma_j^2}\right.\\
			& \qquad \left.-\frac{2\xi_j\eta_j}{\sigma_j^2}\ln\left(\cosh\frac{\gamma_j s}{2}+\frac{\xi_j+\Lambda_j-iu_j\rho_{wz_j}\sigma_j}{\gamma}\sinh\frac{\gamma_j s}{2}\right)\right\}.
\end{split}
\end{align}
\end{lem}

\begin{proof}
We first determine the characteristic function of the continuous parts of the log-prices. Define the function $f$ as
\begin{equation}
f(t,x_1,x_2,v_1,v_2) = \E_\Q\left\{e^{i(u_1 X_{1,T}^c+u_2 X_{2,T}^c)}|X_{1,t}^c = x_1, X_{2,t}^c = x_2, v_{1,t} = v_1, v_{2,t} = v_2\right\},
\end{equation}
and let $\frakM_t = f(t,X_{1,t}^c,X_{2,t}^c,v_{1,t},v_{2,t})$. In the next calculation, we do not yet invoke the assumption that $\rho_w=\rho_z=0$. By It\^{o}'s Lemma, $\frakM_t$ satisfies
\begin{align*}
\dif\frakM_t
	& = \left\{\pder[f]{t}+\sum_{i=1}^2\left(r-q_i-\tilde{\lambda}_i\tilde{\kappa}_i-\frac{1}{2}v_{i,t}\right)\pder[f]{x_i}+\sum_{i=1}^2\left[\xi_i\eta_i-(\xi_i+\Lambda_i)v_{i,t}\right]\pder[f]{v_i}\right.\\
	& \qquad + \frac{1}{2}\sum_{i=1}^2 v_{i,t}\pder[^2f]{x_i^2} + \frac{1}{2}\sum_{i=1}^2\sigma_i^2 v_{i,t}\pder[^2f]{v_i^2}+\rho_w\sqrt{v_{1,t}v_{2,t}}\pder[^2f]{x_1\partial x_2}\\
	& \qquad + \left.\sum_{i=1}^2\rho_{wz_i}\sigma_iv_{i,t}\pder[^2f]{x_i\partial v_i}+\rho_z\sigma_1\sigma_2\sqrt{v_{1,t}v_{2,t}}\pder[^2f]{v_1\partial v_2}\right\}\dif t\\
	& \qquad + \sum_{i=1}^2\sqrt{v_{i,t}}\pder[f]{x_i}\dif\tilde{W}_{i,t}+\sum_{i=1}^2\sigma_i\sqrt{v_{i,t}}\pder[f]{v_i}\dif\tilde{Z}_{i,t}.
\end{align*}
Since $\frakM_t$ is a martingale, the drift coefficient is equal to zero, which leads to the equation
\begin{align*}
0	& = \pder[f]{t}+\sum_{i=1}^2\left(r-q_i-\tilde{\lambda}_i\tilde{\kappa}_i-\frac{1}{2}v_{i,t}\right)\pder[f]{x_i}+\sum_{i=1}^2\left[\xi_i\eta_i-(\xi_i+\Lambda_i)v_{i,t}\right]\pder[f]{v_i}\\
	& \qquad + \frac{1}{2}\sum_{i=1}^2 v_{i,t}\pder[^2f]{x_i^2} + \frac{1}{2}\sum_{i=1}^2\sigma_i^2 v_{i,t}\pder[^2f]{v_i^2}+\rho_w\sqrt{v_{1,t}v_{2,t}}\pder[^2f]{x_1\partial x_2}\\
	& \qquad + \sum_{i=1}^2\rho_{wz_i}\sigma_iv_{i,t}\pder[^2f]{x_i\partial v_i}+\rho_z\sigma_1\sigma_2\sqrt{v_{1,t}v_{2,t}}\pder[^2f]{v_1\partial v_2}.
\end{align*}
The corresponding terminal condition is given by $$f(T,x_1,x_2,v_1,v_2) = e^{i(u_1 x_1+u_2 x_2)}.$$ This PDE is to be solved for $t\in[0,T]$, $(x_1,x_2)\in\mathbb{R}^2$, and $(v_1,v_2)\in\mathbb{R}_+^2$, where $\mathbb{R}_+^2 = (0,\infty)\times(0,\infty)$.

To be able to solve the above PDE explicitly, we now assume that $\rho_w=\rho_z=0$. \citet{CaneOlivares-2014} refer to this situation as the independent volatility case. This simplifies the preceding PDE such that its coefficients become linear. The resulting equation given the the simplifying assumptions is
\begin{align}
0	& = \pder[f]{t}+\sum_{i=1}^2\left(r-q_i-\tilde{\lambda}_i\tilde{\kappa}_i-\frac{1}{2}v_{i,t}\right)\pder[f]{x_i}+\sum_{i=1}^2\left[\xi_i\eta_i-(\xi_i+\Lambda_i)v_{i,t}\right]\pder[f]{v_i}\\
	& \qquad + \frac{1}{2}\sum_{i=1}^2 v_{i,t}\pder[^2f]{x_i^2} + \frac{1}{2}\sum_{i=1}^2\sigma_i^2 v_{i,t}\pder[^2f]{v_i^2}+ \sum_{i=1}^2\rho_{wz_i}\sigma_iv_{i,t}\pder[^2f]{x_i\partial v_i},
\end{align}
and to solve the equation, we guess a solution of the form
\begin{equation}
f(t,x_1,x_2,v_1,v_2) = \exp\left\{i(u_1 x_1+u_2 x_2)+C(T-t)+D_1(T-t)v_1+D_2(T-t)v_2\right\},
\end{equation}
for some functions $C(s)$, $D_1(s)$, and $D_2(s)$ of one variable evaluated at $s=T-t$ \citep{Heston-1993, ContTankov-2004, CaneOlivares-2014}. At this point, we suppress the $t$ subscript and introduce the notation $f_t = f(t,x_1,x_2,v_1,v_2)$. With this specification, the PDE becomes
\begin{align*}
0	& = f_t\left[-C'(T-t)-D_1'(T-t)v_1-D_2'(T-t)v_2\right]\\
	& \qquad+\left(r-q_1-\tilde{\lambda}_1\tilde{\kappa}_1-\frac{1}{2}v_1\right)f_t iu_1+\left(r-q_2-\tilde{\lambda}_2\tilde{\kappa}_2-\frac{1}{2}v_2\right)f_t iu_2\\
	& \qquad+[\xi_1\eta_1-(\xi_1+\Lambda_1)v_1]f_t D_1(T-t)+[\xi_2\eta_2-(\xi_2+\Lambda_2)v_2]f_t D_2(T-t)\\
	& \qquad-\frac{1}{2}v_1 u_1^2 f_t-\frac{1}{2}v_2 u_2^2 f_t +\frac{1}{2}\sigma_1^2 v_1 D_1^2(T-t)f_t+\frac{1}{2}\sigma_2^2 v_2 D_2^2(T-t)f_t\\
	& \qquad+\rho_{wz_1}\sigma_1 v_1 iu_1 D_1(T-t)f_t +\rho_{wz_2}\sigma_2 v_2 iu_2 D_2(T-t)f_t,
\end{align*}
where $C'(T-t)$, $D_1'(T-t)$, and $D_2'(T-t)$ are the first derivatives of $C(s)$, $D_1(s)$, and $D_2(s)$ evaluated at $s=T-t$. Simplifying and collecting the coefficients of $v_1$ and $v_2$, we have

{\small
\begin{align*}
0	& = v_1\left\{-D_1'(T-t)-\frac{iu_1}{2}-\frac{u_1^2}{2}+\frac{1}{2}\sigma_1^2 D_1^2(T-t)+(iu_1\rho_{wz_1}\sigma_1-\xi_1-\Lambda_1)D_1(T-t)\right\}\\
	& \qquad + v_2\left\{-D_2'(T-t)-\frac{iu_2}{2}-\frac{u_2^2}{2}+\frac{1}{2}\sigma_2^2 D_2^2(T-t)+(iu_2\rho_{wz_2}\sigma_2-\xi_2-\Lambda_2)D_2(T-t)\right\}\\
	& \qquad - C'(T-t)+(r-q_1-\tilde{\lambda}_1\tilde{\kappa}_1)iu_1+(r-q_2-\tilde{\lambda}_2\tilde{\kappa}_2)iu_2+\xi_1\eta_1 D_1(T-t)+\xi_2\eta_2 D_2(T-t).
\end{align*}}

\noindent Since $v_1$ and $v_2$ are nonzero, it follows that $C(s)$, $D_1(s)$, and $D_2(s)$ must satisfy the equations
\begin{align}
D_1'(s)	& = \frac{\sigma_1^2}{2}D_1^2(s)+(iu_1\rho_{wz_1}\sigma_1-\xi_1-\Lambda_1)D_1(s)-\frac{1}{2}(iu_1+u_1^2)\\
D_2'(s)	& = \frac{\sigma_2^2}{2}D_2^2(s)+(iu_2\rho_{wz_2}\sigma_2-\xi_2-\Lambda_2)D_2(s)-\frac{1}{2}(iu_2+u_2^2)\\
C'(s)		& = (r-q_1-\tilde{\lambda}_1\tilde{\kappa}_1)iu_1+(r-q_2-\tilde{\lambda}_2\tilde{\kappa}_2)iu_2+\xi_1\eta_1 D_1(s)+\xi_2\eta_2 D_2(s),
\end{align}
subject to the initial condition $D_1(0)=D_2(0)=C(0)=0$. Adapting the results of \citet{Heston-1993}, \citet{Bates-1996}, and \citet{ContTankov-2004}, $D_1(s)$ and $D_2(s)$ are given by
\begin{equation*}
D_j(s) = -\frac{u_j^2+iu_j}{\gamma_j\coth(\gamma_j s/2)+(\xi_j+\lambda_j)-iu_j\rho_{wz_j}\sigma_j}, \qquad j=1,2,
\end{equation*}
where $\gamma_j$ is given by $$\gamma_j = \sqrt{\sigma_j^2(u_j^2+iu_j)+(\xi_j+\Lambda_j-iu_j\rho_{wz_j}\sigma_j)^2}, \qquad j=1,2.$$ Integration for $C(s)$ yields
\begin{align*}
C(s) 	& = \sum_{j=1}^2\left\{iu_j s(r-q_j-\tilde{\lambda}_j\tilde{\kappa}_j)+\frac{\xi_j\eta_j s (\xi_j+\lambda_j-iu_j\rho_{wz_j}\sigma_j)}{\sigma_j^2}\right.\\
			& \qquad \left.-\frac{2\xi_j\eta_j}{\sigma_j^2}\ln\left(\cosh\frac{\gamma_j s}{2}+\frac{\xi_j+\Lambda_j-iu_j\rho_{wz_j}\sigma_j}{\gamma}\sinh\frac{\gamma_j s}{2}\right)\right\}.
\end{align*}

Note that 
\begin{align*}
f(0,x_1,x_2,v_1,v_2) 
	& = \E_\Q\left[\left.e^{i(u_1 X_{1,T}^c+u_2 X_{2,T}^c)}\right|X_{1,0}^c = x_1, X_{2,0}^c = x_2, v_{1,0} = v_1, v_{2,0}=v_2\right]\\
	& = \E_\Q\left[e^{i(u_1 X_{1,T}^c+u_2 X_{2,T}^c)}\right].
\end{align*}
If $x_1$ and $x_2$ represent the log of initial stock prices and $v_1$ and $v_2$ represent initial volatility levels, then from the assumed functional form of $f$, it follows that
\begin{align*}
\E_\Q\left[e^{i(u_1 X_{1,T}^c+u_2 X_{2,T}^c)}\right]
	& = \exp\left\{i(u_1 x_1+u_2 x_2)+C(T)+D_1(T)v_1+D_2(T)v_2\right\},
\end{align*} 
where $C(T)$, $D_1(T)$, and $D_2(T)$ are given as above. Thus for any $0<t\leq T$, the joint characteristic function of the continuous part of the log-asset prices, which we denote by $\phi_{X_t^c}(u_1,u_2)$ is given by
\begin{equation}
\phi_{X_t^c}(u_1,u_2) = \exp\left\{i(u_1 x_1+u_2 x_2)+C(t;u_1,u_2)+D_1(t;u_1,u_2)v_1+D_2(t;u_1,u_2)v_2\right\}
\end{equation}

We now turn to the joint characteristic function $\phi_{Q_t}(u_1,u_2)$ of the jump parts. From the assumptions on the SVJD model, we note that the jump components of stock 1 and stock 2 are independent, hence
\begin{align*}
\phi_{Q_t}(u_1,u_2)
& = \E_\Q\left[\exp\left(iu_1\sum_{n=1}^{N_{1,t}}Y_{1,n}+iu_2\sum_{n=1}^{N_{2,t}}Y_{2,n}\right)\right]\\
& = \E_\Q\left[\exp\left(iu_1\sum_{n=1}^{N_{1,t}}Y_{1,n}\right)\right\}\E_\Q\left\{\exp\left(iu_2\sum_{n=1}^{N_{2,t}}Y_{2,n}\right)\right]
\end{align*}
If $\phi_{Y_i}(\cdot)$ denotes the common characteristic function of the $Y_{i,n}$'s under the risk-neutral measure $\Q$, then we have
\begin{equation}
\phi_{Q_t}(u_1,u_2) = \exp\left\{\tilde{\lambda}_1 t\left(\phi_{Y_1}(u_1)-1\right)+\tilde{\lambda}_2 t\left(\phi_{Y_2}(u_2)-1\right)\right\},
\end{equation}
where $\tilde{\lambda}_1$ and $\tilde{\lambda}_2$ are the intensities of $N_{1,t}$ and $N_{2,t}$, respectively, under $\Q$.

Following the independence of the continuous and jump parts of the log-price processes, the joint characteristic function of the log-prices $X_{1,t}$ and $X_{2,t}$, which we denote by $\phi_{X_t}(u_1,u_2)$ is therefore given by
\begin{align*}
\phi_{X_t}(u_1,u_2)
	& = \exp\left\{i(u_1 x_1+u_2 x_2)+C(t;u_1,u_2)+D_1(t;u_1,u_2)v_1+D_2(t;u_1,u_2)v_2\right\}\\
	& \qquad\times \exp\left\{\tilde{\lambda}_1 t\left(\phi_{Y_1}(u_1)-1\right)+\tilde{\lambda}_2 t\left(\phi_{Y_2}(u_2)-1\right)\right\}.
\end{align*}
This is defined for all $(x_1,x_2)\in\mathbb{R}^2$, $(v_1,v_2)\in\mathbb{R}^2_+$, $t\in[0,T]$, and $u_1,u_2\in\mathbb{C}$.
\end{proof}

\begin{rem}
Let $\phi_{X_T|t}(u_1,u_2) = \E_\Q\left[\left.e^{i(u_1 X_{1,T}+u_2 X_{2,T})}\right|\calF_t\right]$. From the above calculations, we have
\begin{align}
\begin{split}
\label{eqn-SVJD-LogPriceCF-Conditional}
& \phi_{X_T|t}(u_1,u_2)\\
& \qquad = \exp\left\{i(u_1 x_1+u_2 x_2)+C(T-t;u_1,u_2)+D_1(T-t;u_1,u_2)v_1+D_2(T-t;u_1,u_2)v_2\right\}\\
& \qquad \qquad \times \exp\left\{\tilde{\lambda}_1 (T-t)\left(\phi_{Y_1}(u_1)-1\right)+\tilde{\lambda}_2 (T-t)\left(\phi_{Y_2}(u_2)-1\right)\right\},
\end{split}
\end{align}
where $x_1$, $x_2$, $v_1$, and $v_2$ now denote time $t$ values of the processes $\{X_{1,t}\}$, $\{X_{2,t}\}$, $\{v_{1,t}\}$, and $\{v_{2,t}\}$, respectively.
\end{rem}

At this point, we briefly state the results of \citet{CaldanaFusai-2013} on the approximate pricing of European spread options under general stock price dynamics, as this will be used to derive European exchange option prices under the assumed SVJD model. Consider a European option with strike price $K$ written on the spread $S_{1,T}-S_{2,T}$.
Then the price today of the spread option is $C_{K,0} = e^{-rT}\E_\Q\left[(S_{1,T}-S_{2,T}-K)^+\right].$ For some parameters $\alpha$ and $k$, define the event $$A = \left\{S_{1,T}>e^k S_{2,T}/\E_\Q[S_{2,T}^\alpha]\right\}.$$ This event, as noted by \citet{Bjerskund-2011}, is a sub-optimal exercise strategy and produces the following lower bound on the European spread option payoff,
\begin{equation}
(S_{1,T}-S_{2,T}-K)^+ \geq (S_{1,T}-S_{2,T}-K)\vm{1}_A.
\end{equation}
It follows therefore that a lower bound for the spread option price is given by
\begin{equation}
C^{k,\alpha}_{K,0} = e^{-rT}\E_\Q\left[(S_{1,T}-S_{2,T}-K)\vm{1}_A\right].
\end{equation}
The following theorem \citep[Proposition 1]{CaldanaFusai-2013} provides an expression for $C^{k,\alpha}_{K,0}$ in terms of the joint characteristic function $\phi_{X_T}(u_1,u_2)$of the log-price vector $X_T = (X_{1,T},X_{2,T})^\top$.

\begin{thm}
\label{thm-CaldanaFusaiApproximation}
The lower bound $C^{k,\alpha}_{K,0}$ for the European spread option price is given by
\begin{equation}
C^{k,\alpha}_{K,0} = \left\{\frac{e^{-\delta k-rT}}{\pi}\int_0^\infty e^{-izk}\Psi_T(z,\delta,\alpha)\dif z\right\}^+,
\end{equation}
where
\begin{align}
\begin{split}
\psi_T(z,\delta,\alpha)
	& = \frac{e^{i(z-i\delta)\ln\phi_{X_T}(0,-i\alpha)}}{i(z-i\delta)}\left[\phi_{X_T}\left(z-i\delta-i,-\alpha(z-i\delta)\right)\right.\\
	& \qquad \left.-\phi_{X_T}\left(z-i\delta,-\alpha(z-i\delta)-i\right)-K\phi_{X_T}\left(z-i\delta,-\alpha(z-i\delta)\right)\right],
\end{split}
\end{align}
and
\begin{equation}
\alpha = \frac{\E_\Q[S_{2,T}]}{\E_\Q[S_{2,T}]+K}, \qquad k = \ln\left(\E_\Q[S_{2,T}]+K\right).
\end{equation}
\end{thm}


The preceding theorem requires the existence of the joint characteristic function of the log-prices of the two stocks. In our model, we have shown that the characteristic function $\phi_{X_t}$ is given by equation \eqref{eqn-SVJD-LogPriceCF}, and hence the above result can be applied in approximating European spread option prices and pricing the European exchange option under our SVJD model.

We make some remarks on the other quantities that appear in the formula. First, $\delta$ is associated to the exponentially decaying term $e^{-\delta k}$ that must be multiplied to $C^{k,\alpha}_{K,0}$ to produce a square-integrable term in the negative $k$-axis. A discussion of the choice of $\delta$ can be found in \citet{CarrMadan-1999} and \citet{DempsterHong-2002}. Next, since the discounted yield process $\{e^{-(r-q_2)t}S_{2,t}\}$ is a martingale under $\Q$, $\E_\Q[S_{2,T}]$ represents the forward price today of asset 2 for delivery at time $T$. In terms of the characteristic function, the forward price can also be expressed as $\phi_{X_T}(0,-i)$. The positive part $(\cdot)^+$ is required since, without it, the formula (as well as the original \citet{Bjerskund-2011} result) may produce negative values for deeply out-of-the-money options. Thus for practical purposes, out-of-the-money exchange options are assigned a value of 0 \citep{CaldanaFusai-2013}. Note that this consideration is consistent with equation \eqref{eqn-EuExcOp-Price1}, since if the option is out-of-the-money, then $\vm{1}_\calA=0$, which results to $C_0^E=0$.

\citet{CaldanaFusai-2013} note that this result is an improvement from the \citet{HurdZhou-2010} Fourier inversion formula since the exchange option case ($K=0$) can be handled here without complications. In this regard, the \citet{CaldanaFusai-2013} formula also gives an \emph{exact} price for the exchange option through an appropriate choice of the parameters $k$ and $\alpha$ (which will be discussed next). Furthermore, the integration in the formula above involves a univariate Fourier inversion in contrast to the bivariate inversion of \citet{HurdZhou-2010}, implying that the \citet{CaldanaFusai-2013} result requires less computation time. However, the result of \citet{CaldanaFusai-2013} is a lower bound approximation for the spread option price, in contrast to the exact price derived by \citet{HurdZhou-2010}. The \citet{HurdZhou-2010} formula also does not depend on the decay parameter $\delta$.

In the following proposition we extend the \citet{CaldanaFusai-2013} result in Theorem \ref{thm-CaldanaFusaiApproximation} to a version that provides a lower bound for the European spread option price at any time $t\in[0,T)$.

\begin{prop}
\label{prop-CaldanaFusaiApproximationExtension}
A lower bound for the time $t\in[0,T)$ price of a European spread option with strike price $K$ is given by
\begin{equation}
C_{K,t}^{\alpha,k} = \left\{\frac{e^{-r(T-t)-\delta k}}{\pi}\int_0^\infty e^{-izk}\psi_{T|t}(z,\delta,\alpha)\dif z\right\}^+,
\end{equation}
where
\begin{align}
\begin{split}
\psi_{T|t}(z,\delta,\alpha)
	& = \frac{e^{i(z-i\delta)\ln\phi_{X_T|t}(0,-i\alpha)}}{iz+\delta}\left[\phi_{X_T|t}\left(z-i\delta-i,-\alpha(z-i\delta)\right)\right.\\
	& \qquad \left.\phi_{X_T|t}\left(z-i\delta,-\alpha(z-i\delta)-i\right)-K\phi_{X_T|t}\left(z-i\delta,-\alpha(z-i\delta)\right)\right],
\end{split}
\end{align}
$\phi_{X_T|t}=\E_\Q[e^{iu_1 X_{1,T}+iu_2 X_{2,T}}|\calF_t]$, $\alpha$ and $k$ are given by
\begin{equation}
\alpha = \frac{F_2(t,T)}{F_2(t,T)+K} \qquad \text{and} \qquad k=\ln\left[F_2(t,T)+K\right],
\end{equation}
and $F_2(t,T)$ is the time $t$ forward price of asset with for delivery at time $T$.
\end{prop} 

\begin{proof}
The proof we present follows the outline of Theorem \ref{thm-CaldanaFusaiApproximation} in Appendix A of \citet{CaldanaFusai-2013}. Recall that the time $t$ price of the European spread option is given by $$e^{-r(T-t)}\E_\Q\left[\left.(S_{1,T}-S_{2,T}-K)^+\right|\calF_t\right].$$ To obtain an approximation for the time $t$ price,  we define the event $A_t$ as
\begin{equation}
A_t = \left\{S_{1,T}\geq\frac{e^k S_{2,T}^\alpha}{\E_\Q[S_{2,T}^\alpha|\calF_t]}\right\},
\end{equation}
where $\alpha$ and $k$ are defined as above. This is similar to the event $A$ defined by \citet{CaldanaFusai-2013} (see discussion preceding Theorem \ref{thm-CaldanaFusaiApproximation}), except that the expectation is conditional on $\calF_t$ and forward prices are taken at time $t$. Following the argument of \citet{Bjerskund-2011}, the quantity
\begin{equation}
C_{K,t}^{\alpha,k} = e^{-r(T-t)}\E_\Q\left[\left.(S_{1,T}-S_{2,T}-K)\vm{1}_{A_t}\right|\calF_t\right]
\end{equation}
is a lower bound for the true option price at time $t$.

At this point, we rewrite some of the quantities above in terms of the notation established before. First, we note that since the discounted yield process $\{e^{-(r-q_2)t}S_{2,t}\}$ is a $\Q$-martingale, the forward price $F_2(t,T)$ can be expressed as $\E_\Q[S_{2,T}|\calF_t]$. In terms of the conditional characteristic function $\phi_{X_T|t}$ (which, under the SVJD model, is given by equation \eqref{eqn-SVJD-LogPriceCF-Conditional}), the quantity $\E_\Q[S_{2,T}^\alpha|\calF_t]$ is given by $\phi_{X_T|t}(0,-i\alpha)$. Furthermore, we may rewrite the exercise strategy $A_t$ in terms of log-prices as
\begin{equation}
A_t = \left\{X_{1,T}-\alpha X_{2,T}+\ln\phi_{X_T|t}(0,-i\alpha)>k\right\}.
\end{equation}

Now, we seek to express $C_K^{\alpha,k}(t)$ as a Fourier inversion formula following the proof in \citet{CaldanaFusai-2013} (see Appendix A of their paper). To this end, let $\delta$ be a positive number and define $\psi_{T|t}$ as $$\psi_{T|t}(z,\delta,\alpha) = \int_{-\infty}^{\infty} e^{izk+\delta k}\E_\Q\left[\left.(S_{1,T}-S_{2,T}-K)\vm{1}_{A_t}\right|\calF_t\right]\dif k.$$ In other words, $\psi_{T|t}$ is the Fourier transform of $e^{\delta k}\E_\Q[(S_{1,T}-S_{2,T}-K)\vm{1}_{A_t}|\calF_t]$ in the $k$-variable. Let $g(x_1,x_2)$ be the transition density function of the log-price vector $X_T$. Thus, the $\psi_{T|t}$ may be evaluated as the triple integral $$\psi_{T|t}(z,\delta,\alpha) = \int_{-\infty}^\infty e^{izk+\delta k}\left[\int_{-\infty}^\infty\int_\beta^\infty \left(e^{x_1}-e^{x_2}-K\right)g(x_1,x_2)\dif x_1\dif x_2\right]\dif k,$$ where $\beta = k+\alpha x_2-\ln\phi_{X_T|t}(0,-i\alpha)$. Let $\beta' = x_1-\alpha x_2+\ln\phi_{X_T|t}(0,-i\alpha)$. Then, after changing the order of integration, we have the following:
\begin{align*}
\psi_{T|t}(z,\delta,\alpha)
	& = \int_{-\infty}^\infty\int_{-\infty}^\infty\left[\int_{-\infty}^{\beta'}e^{izk+\delta k}\dif k\right]\left(e^{x_1}-e^{x_2}-K\right)g(x_1,x_2)\dif x_1\dif x_2\\
	& = \int_{-\infty}^\infty\int_{-\infty}^\infty\frac{e^{i(z-i\delta)\beta'}}{iz+\delta}\left(e^{x_1}-e^{x_2}-K\right)g(x_1,x_2)\dif x_1\dif x_2\\
	& = \frac{e^{i(z-i\delta)\ln\phi_{X_T|t}(0,-i\alpha)}}{iz+\delta}\\
	& \qquad \times \int_{-\infty}^\infty\int_{-\infty}^\infty e^{i(z-i\delta)(x_1-\alpha x_2)}\left(e^{x_1}-e^{x_2}-K\right)g(x_1,x_2)\dif x_1\dif x_2\\
	& = \frac{e^{i(z-i\delta)\ln\phi_{X_T|t}(0,-i\alpha)}}{iz+\delta}\\
	& \qquad \left[\int_{-\infty}^\infty\int_{-\infty}^\infty \exp\left\{i(z-i\delta-i)x_1-i\alpha(z-i\delta)x_2\right\}g(x_1,x_2)\dif x_1\dif x_2\right.\\
	& \qquad - \int_{-\infty}^\infty\int_{-\infty}^\infty \exp\left\{i(z-i\delta)x_1-i(\alpha z-i\alpha\delta+i)x_2\right\}g(x_1,x_2)\dif x_1\dif x_2\\
	& \qquad - K\left.\int_{-\infty}^\infty\int_{-\infty}^\infty \exp\left\{i(z-i\delta)x_1-i\alpha(z-i\delta)x_2\right\}g(x_1,x_2)\dif x_1\dif x_2\right]
\end{align*}
The preceding double integrals may be written in terms of $\phi_{X_T|t}$, which thus results to
\begin{align*}
\psi_{T|t}(z,\delta,\alpha)
	& = \frac{e^{i(z-i\delta)\ln\phi_{X_T|t}(0,-i\alpha)}}{iz+\delta}\left[\phi_{X_T|t}\left(z-i\delta-i,-\alpha(z-i\delta)\right)\right.\\
	& \qquad \left.\phi_{X_T|t}\left(z-i\delta,-\alpha(z-i\delta)-i\right)-K\phi_{X_T|t}\left(z-i\delta,-\alpha(z-i\delta)\right)\right]
\end{align*}
We then invert the Fourier transform to obtain $$e^{\delta k}\E_\Q[(S_{1,T}-S_{2,T}-K)\vm{1}_{A_t}|\calF_t] = \frac{1}{\pi}\int_0^\infty e^{-izk}\psi_{T|t}(z,\delta,\alpha)\dif z.$$ It follows therefore that $$C_K^{\alpha,k}(t) = \left\{\frac{e^{-r(T-t)-\delta k}}{\pi}\int_0^\infty e^{-izk}\psi_{T|t}(z,\delta,\alpha)\dif z\right\}^+.$$ We require the positive part in the preceding formula to avoid negative prices for deeply-out-of-money options \citep{CaldanaFusai-2013}. As in \citet{Bjerskund-2011} and \citet{CaldanaFusai-2013} $\alpha$ and $k$ may be chosen as $$\alpha = \frac{F_2(t,T)}{F_2(t,T)+K} \qquad \text{and} \qquad k = \ln\left[F_2(t,T)+K\right].$$
\end{proof}

\begin{rem}
Theorem \ref{thm-CaldanaFusaiApproximation} follows naturally by setting $t=0$.
\end{rem}


We recall the approximate strategy $A_t$ introduced in the proof of Proposition \ref{prop-CaldanaFusaiApproximationExtension}. If we set $K=0$ (as is the case for the exchange option) in the expressions for $\alpha$ and $k$, we obtain $\alpha=1$ and $k = \ln F_2(t,T)$. Since $F_2(t,T)=\E_\Q[S_{2,T}|\calF_t]$, setting $K=0$ causes $A_t$ to coincide with the true exercise strategy $B=\{S_{1,T}\geq S_{2,T}\}$ for the European exchange option. Thus, similar to the findings of \citet{CaldanaFusai-2013}, $C_{K,t}^{\alpha,k}$ (or $C_{K,0}^{\alpha,k}$ in Theorem \ref{thm-CaldanaFusaiApproximation}) is exact for the European exchange option. The following corollary provides the time $t$ price of the European exchange option.

\begin{cor}
\label{cor-EuExcOp-FTPrice}
The time $t$ price of the European exchange option is given by
\begin{equation}
\label{eqn-EuExcOp-FTPrice}
C_t^E = \left\{\frac{e^{-r(T-t)-\delta k}}{\pi}\int_0^\infty e^{-izk}\psi_{T|t}(z,\delta,1)\dif z\right\}^+,
\end{equation}
where
\begin{align*}
\psi_{T|t}(z,\delta,1)
	& = \frac{e^{i(z-i\delta)k}}{iz+\delta}\left[\phi_{X_T|t}\left(z-i\delta-i,-z+i\delta\right)-\phi_{X_T|t}\left(z-i\delta,-z+i\delta-i\right)\right].
\end{align*}
\end{cor}



In the analysis that follows, we focus on the price of the European exchange option at time 0. Suppose the exchange option is \emph{not} deeply out-of-the-money (i.e. $C_0^E\geq 0$). Then the following proposition expresses the \citet{CaldanaFusai-2013} result into a form that is consistent with what we obtained in equation \eqref{eqn-EuExcOp-Price-Decomposed} from the change-of-num\'eraire technique. 

\begin{prop}
\label{prop-EuExcOp-FTPrice-Decomposed}
The time $0$ price of the European exchange option is given by
\begin{align}
\begin{split}
\label{eqn-EuExcOp-FTPrice-Decomposed}
C_0^E
	& = S_{1,0}e^{-q_1 T}\int_0^\infty\frac{1}{\pi(iz+\delta)}\phi_{X_T}^{(1)}(iz+\delta,-iz-\delta)\dif z\\
	& \qquad -S_{2,0}e^{-q_2 T}\int_0^\infty\frac{1}{\pi(iz+\delta)}\phi_{X_T}^{(2)}(iz+\delta,-iz-\delta)\dif z.
\end{split}
\end{align}
where $\phi_{X_T}^{(1)}$ and $\phi_{X_T}^{(2)}$ denote the joint characteristic function of $X_T = (X_{1,T},X_{2,T})^\top$ under $\hat{\Q}_1$ and $\hat{\Q}_2$ (defined in Section \ref{sec-ChangeofNumeraire}), respectively.
\end{prop}

\begin{proof}
Since we assume that the option price is nonnegative, we may remove the positive part in equation \eqref{eqn-EuExcOp-FTPrice} (with $t=0$) and write $$C_0^E = e^{-rT}\int_0^\infty\frac{e^{-(iz+\delta)k}}{\pi}\psi_T(z,\delta,1)\dif z.$$ Substituting the expression for $\psi_T$, we have
\begin{align*}
C_0^E
	& = e^{-rT}\left\{\int_0^\infty\frac{e^{-(iz+\delta)k+(iz+\delta)k}}{\pi(iz+\delta)}\phi_{X_T}\left(z-i\delta-i,-z+i\delta\right)\dif z\right.\\
	& \qquad \left.-\int_0^\infty\frac{e^{-(iz+\delta)k+(iz+\delta)k}}{\pi(iz+\delta)}\phi_{X_T}\left(z-i\delta,-z+i\delta-i\right)\dif z\right\}.
\end{align*}
We also note that
\begin{align*}
\phi_{X_T}(z-i\delta-i,-z+i\delta) & = \E_\Q\left[S_{1,T}\exp\left\{(iz+\delta)X_{1,T}+(-iz-\delta)X_{2,T}\right\}\right]\\
\phi_{X_T}(z-i\delta,-z+i\delta-i) & = \E_\Q\left[S_{2,T}\exp\left\{(iz+\delta)X_{1,T}+(-iz-\delta)X_{2,T}\right\}\right].
\end{align*}
But note that for $i=1,2$, the terminal stock prices can be written as $S_{i,T} = S_{i,0}e^{(r-q_i)T}U_{i,T}$, where $U_{i,T}$ is defined as in equation \eqref{eqn-EuExcOp-RNDerivative}. We can thus write
\begin{align*}
\phi_{X_T}(z-i\delta-i,-z+i\delta) & = S_{1,0}e^{(r-q_1)T}\E_\Q\left[U_{1,T}\exp\left\{(iz+\delta)X_{1,T}+(-iz-\delta)X_{2,T}\right\}\right]\\
\phi_{X_T}(z-i\delta,-z+i\delta-i) & = S_{2,0}e^{(r-q_2)T}\E_\Q\left[U_{2,T}\exp\left\{(iz+\delta)X_{1,T}+(-iz-\delta)X_{2,T}\right\}\right].
\end{align*}
Using these expressions for $\phi_{X_T}$, the European exchange option price is therefore given by
\begin{align*}
C_0^E
	& = S_{1,0}e^{-q_1 T}\int_0^\infty\frac{1}{\pi(iz+\delta)}\E_\Q\left[U_{1,T}\exp\left\{(iz+\delta)X_{1,T}+(-iz-\delta)X_{2,T}\right\}\right]\dif z\\
	& \qquad -S_{2,0}e^{-q_2 T}\int_0^\infty\frac{1}{\pi(iz+\delta)}\E_\Q\left[U_{2,T}\exp\left\{(iz+\delta)X_{1,T}+(-iz-\delta)X_{2,T}\right\}\right]\dif z.
\end{align*}
Recall from Section \ref{sec-ChangeofNumeraire} that $U_{1,T}$ and $U_{2,T}$ are Radon-Nikod\'ym derivatives that define new probability measures $\hat{\Q}_1$ and $\hat{\Q}_2$ equivalent to $\Q$. Therefore, the expectations that appear above define joint characteristic functions of the log-prices under these new probability measures. The desired result thus follows from replacing the expectations with $\phi_{X_T}^{(1)}$ and $\phi_{X_T}^{(2)}$, the characteristic function of $X_T = (X_{1,T},X_{2,T})^\top$ under $\hat{\Q}_1$ and $\hat{\Q}_2$ respectively.
\end{proof}

\begin{rem}
Under the SVJD specification, the joint characteristic functions $\phi_{X_T}^{(1)}$ and $\phi_{X_T}^{(2)}$ have forms similar to $\phi_{X_T}$ as provided in equation \eqref{eqn-SVJD-LogPriceCF}, with some slight changes in the values of the parameters due to the change of measure. Properties of log-prices after the change of measure follow from the discussion in Section \ref{sec-ChangeofNumeraire}, under the assumption that $\rho_w=\rho_z=0$ in addition to Assumption \ref{asp-ParameterAssumptions}. As such, $\phi_{X_T}^{(1)}$ and $\phi_{X_T}^{(2)}$ may be determined in a manner similar to the calculations presented in the first part of this section.
\end{rem}




The above calculations therefore show that the \citet{CaldanaFusai-2013} result, when applied to European exchange options, allows for a decomposition similar to equation \eqref{eqn-EuExcOp-Price-Decomposed}, which was obtained via the change-of-num\'eraire technique. This therefore presents the possibility that the probabilities $\hat{\Q}_1(\calA_0)$ and $\hat{\Q}_2(\calA_0)$ (the probability of the option being in-the-money under the alternative measures $\hat{\Q}_1$ and $\hat{\Q}_2$) may be computed using Fourier inversion in equation \eqref{eqn-EuExcOp-FTPrice-Decomposed}.


The same analysis can be applied to produce a similar representation for the time $t$ price of the European exchange option, as shown in the next proposition. 

\begin{prop}
The time $t$ European exchange option price is given by
\begin{align}
\begin{split}
\label{eqn-EuExcOp-FTPrice-Decomposed-Timet}
C_t^E & = S_{1,t}e^{-q_1(T-t)}\int_0^\infty\frac{1}{\pi(iz+\delta)}\phi_{X_T|t}^{(1)}(iz+\delta,-iz-\delta)\dif z\\
			& \qquad S_{2,T}e^{-q_2(T-t)}\int_0^\infty\frac{1}{\pi(iz+\delta)}\phi_{X_T|t}^{(2)}(iz+\delta,-iz-\delta)\dif z,
\end{split}
\end{align}
where $$\phi_{X_T|t}^{(i)}(u_1,u_2) = \E_{\hat{\Q}_i}\left[\left.e^{iu_1 X_{1,T}+iu_2 X_{2,T}}\right|\calF_t\right], \qquad i=1,2$$ gives the joint conditional characteristic function of $X_T = (X_{1,T},X_{2,T})^\top$ under $\hat{\Q}_i$, and $\hat{\Q}_1$ and $\hat{\Q}_2$ are the probability measures equivalent to $\Q$ discussed in Section \ref{sec-ChangeofNumeraire}.
\end{prop}

\begin{proof}
The proof of this proposition is similar to that of Proposition \ref{prop-EuExcOp-FTPrice-Decomposed}, except that the expectations involved are replaced by their conditional counterparts (see Corollary \ref{cor-EuExcOp-FTPrice}).
\end{proof}

Because of the equivalence between equations \eqref{eqn-EuExcOp-Price-Decomposed-Timet} and \eqref{eqn-EuExcOp-FTPrice-Decomposed-Timet}, it follows that equation \eqref{eqn-EuExcOp-FTPrice-Decomposed-Timet} characterizes a solution to the IPDE in Proposition \ref{prop-SVJD-PricingIPDE} with terminal condition $C_T^E = (S_{1,T}-S_{2,T})^+$.

In summary, the European exchange option price under the SVJD model may be represented by equation \eqref{eqn-EuExcOp-Price-Decomposed}, which is a formula that resembles the original \citet{Margrabe-1978} result under the Black-Scholes framework. It is also notable that our results in equations \eqref{eqn-EuExcOp-Price-Decomposed} and \eqref{eqn-EuExcOp-FTPrice-Decomposed} do not contain the risk-free rate $r$, similar to the \citet{Margrabe-1978} formula. Due to the addition of stochastic volatilities, a closed formula for the probabilities in equation \eqref{eqn-EuExcOp-Price-Decomposed} cannot be obtained. However, by making additional assumptions on the correlation structure of the Wiener processes involved, we were able to obtain equation \eqref{eqn-EuExcOp-FTPrice-Decomposed}, a Fourier inversion formula for the European exchange option price (following the work of \citet{CaldanaFusai-2013} and \citet{CaneOlivares-2014}). Lastly, by extending the analysis of \citet{CaldanaFusai-2013} to obtain time-$t$ option prices, we were able to obtain a solution $C_t^E$ of IPDE \eqref{prop-SVJD-PricingIPDE}, for the case of the European exchange option, in terms of Fourier inversion formulas.

\section{A Representation of the American Exchange Option Price}
\label{sec-AmericanExchangeOption}

In this section, we show that the price of the American exchange option, under our SVJD model for underlying stock prices, can be decomposed into a sum of the European exchange option price and the early exercise premium. As will be shown, the early exercise premium can be decomposed into a premium arising from the diffusion part of the asset dynamics and premia arising from the possibility of sudden jumps in the asset prices. For expositional convenience, we proceed first with the discussion then consolidate our main result in Proposition \ref{prop-SVJD-AmExcOpPriceRep}.

Note that equation \eqref{eqn-QDynamics-ExcOptionPrice} also applies to the discounted American exchange option price $\tilde{C}_t^A$. In terms of the partial differential operator $\calL$, $\tilde{C}_t^A = e^{-rt}C_{t}^A$ satisfies the stochastic differential equation
\begin{align}
\begin{split}
\dif\tilde{C}_t^A
	& = e^{-rt}\left\{\calL[C_{t-}^A]-rC_{t}^A+\tilde{\lambda}_1\E_\Q^{Y_1}\left[C_t^A\left(S_{1,t-}e^{Y_1},S_{2,t-},v_{1,t},v_{2,t}\right)-C_{t-}^A\right]\right.\\
	& \qquad \left.+ \tilde{\lambda}_2\E_\Q^{Y_2}\left[C_t^A\left(S_{1,t},S_{2,t}e^{Y_2},v_{1,t},v_{2,t}\right)-C_{t-}^A\right]\right\}\dif t\\
	& \qquad + e^{-rt}\sum_{i=1}^2\sqrt{v_{i,t}}S_{i,t-}\pder[C_{t-}^A]{s_i}\dif\tilde{W}_{i,t}+e^{-rt}\sum_{i=1}^2 \sigma_i\sqrt{v_{i,t}}\pder[C_{t-}^A]{v_i}\dif\tilde{Z}_{i,t}\\
	& \qquad + e^{-rt}\int_\mathbb{R}\left[C_t\left(S_{1,t-}e^{y_1},S_{2,t-},v_{1,t},v_{2,t}\right)-C_{t-}^A\right]q(\dif y_1,\dif t)\\
	& \qquad + e^{-rt}\int_\mathbb{R}\left[C_t\left(S_{1,t-},S_{2,t-}e^{y_2},v_{1,t},v_{2,t}\right)-C_{t-}^A\right]q(\dif y_2,\dif t),
\end{split}
\end{align}
where $C_{t-}^A = C_t^A(S_{1,t-},S_{2,t-},v_{1,t},v_{2,t})$ represents the pre-jump price of the American exchange option at time $t$. In integral form, we have
\begin{align*}
\tilde{C}_{T}^A
	& = \tilde{C}_t^A+\int_t^T e^{-ru}\left\{\calL[C_{u-}^A]-rC_{u}^A+\tilde{\lambda}_1\E_\Q^{Y_1}\left[C_u^A\left(S_{1,u-}e^{Y_1},S_{2,u-},v_{1,u},v_{2,u}\right)-C_{u-}^A\right]\right.\\
	& \qquad \left.+ \tilde{\lambda}_2\E_\Q^{Y_2}\left[C_u^A\left(S_{1,u},S_{2,u}e^{Y_2},v_{1,u},v_{2,u}\right)-C_{u-}^A\right]\right\}\dif u\\
	& \qquad + \int_t^T e^{-ru}\sum_{i=1}^2\sqrt{v_{i,u}}S_{i,u-}\pder[C_{u-}^A]{s_i}\dif\tilde{W}_{i,u}+\int_t^T e^{-ru}\sum_{i=1}^2 \sigma_i\sqrt{v_{i,u}}\pder[C_{u-}^A]{v_i}\dif\tilde{Z}_{i,u}\\
	& \qquad + \int_t^T e^{-ru}\int_\mathbb{R}\left[C_u\left(S_{1,u-}e^{y_1},S_{2,u-},v_{1,u},v_{2,u}\right)-C_{u-}^A\right]q(\dif y_1,\dif u)\\
	& \qquad + \int_t^T e^{-ru}\int_\mathbb{R}\left[C_u\left(S_{1,u-},S_{2,u-}e^{y_2},v_{1,u},v_{2,u}\right)-C_{u-}^A\right]q(\dif y_2,\dif u).
\end{align*}
A division by $e^{-rt}$ produces the equation
{\footnotesize
\begin{align*}
\frac{C_T^A}{e^{r(T-t)}}
	& = C_t^A+\int_t^T e^{-r(u-t)}\left\{\calL[C_{u-}^A]-rC_{u}^A+\tilde{\lambda}_1\E_\Q^{Y_1}\left[C_u^A\left(S_{1,u-}e^{Y_1},S_{2,u-},v_{1,u},v_{2,u}\right)-C_{u-}^A\right]\right.\\
	& \qquad \left.+ \tilde{\lambda}_2\E_\Q^{Y_2}\left[C_u^A\left(S_{1,u},S_{2,u}e^{Y_2},v_{1,u},v_{2,u}\right)-C_{u-}^A\right]\right\}\dif u\\
	& \qquad + \int_t^T e^{-r(u-t)}\sum_{i=1}^2\sqrt{v_{i,u}}S_{i,u-}\pder[C_{u-}^A]{s_i}\dif\tilde{W}_{i,u}+\int_t^T e^{-r(u-t)}\sum_{i=1}^2 \sigma_i\sqrt{v_{i,u}}\pder[C_{u-}^A]{v_i}\dif\tilde{Z}_{i,u}\\
	& \qquad + \int_t^T e^{-r(u-t)}\int_\mathbb{R}\left[C_u\left(S_{1,u-}e^{y_1},S_{2,u-},v_{1,u},v_{2,u}\right)-C_{u-}^A\right]q(\dif y_1,\dif u)\\
	& \qquad + \int_t^T e^{-r(u-t)}\int_\mathbb{R}\left[C_u\left(S_{1,u-},S_{2,u-}e^{y_2},v_{1,u},v_{2,u}\right)-C_{u-}^A\right]q(\dif y_2,\dif u).
\end{align*}}
To further simplify notation, define the operator $\bar{\calL}$ as
\begin{align}
\begin{split}
\bar{\calL}[C_{u-}^A]
	& = \calL[C_{u-}^A]-rC_{u-}^A+\tilde{\lambda}_1\E_\Q^{Y_1}\left[C_u^A\left(S_{1,u-}e^{Y_1},S_{2,u-},v_{1,u},v_{2,u}\right)-C_{u-}^A\right]\\
	& \qquad + \tilde{\lambda}_2\E_\Q^{Y_2}\left[C_u^A\left(S_{1,u-},S_{2,u-}e^{Y_2},v_{1,u},v_{2,u}\right)-C_{u-}^A\right],
\end{split}
\end{align}
so the equation above becomes
\begin{align*}
\frac{C_T^A}{e^{r(T-t)}}
	& = C_t^A+\int_t^T e^{-r(u-t)}\bar{\calL}[C_{u-}^A]\dif u\\
	& \qquad + \int_t^T e^{-r(u-t)}\sum_{i=1}^2\sqrt{v_{i,u}}S_{i,u-}\pder[C_{u-}^A]{s_i}\dif\tilde{W}_{i,u}+\int_t^T e^{-r(u-t)}\sum_{i=1}^2 \sigma_i\sqrt{v_{i,u}}\pder[C_{u-}^A]{v_i}\dif\tilde{Z}_{i,u}\\
	& \qquad + \int_t^T e^{-r(u-t)}\int_\mathbb{R}\left[C_u\left(S_{1,u-}e^{y_1},S_{2,u-},v_{1,u},v_{2,u}\right)-C_{u-}^A\right]q(\dif y_1,\dif u)\\
	& \qquad + \int_t^T e^{-r(u-t)}\int_\mathbb{R}\left[C_u\left(S_{1,u-},S_{2,u-}e^{y_2},v_{1,u},v_{2,u}\right)-C_{u-}^A\right]q(\dif y_2,\dif u).
\end{align*}

As remarked below equation \eqref{eqn-QDynamics-ExcOptionPrice}, terms involving the $\Q$-Wiener processes and the $\Q$-compensated counting measures are martingales with zero mean under $\Q$. Hence, taking the conditional expectation of the above equation under $\Q$ with respect to $\calF_t$, we obtain
\begin{equation}
\label{eqn-AmericanExcOptionRep1}
\E_\Q\left[\left.e^{-r(T-t)}C_T^A\right|\calF_t\right] = C_t^A + \E_\Q\left[\left.\int_t^T e^{-r(u-t)}\bar{\calL}[C_{u-}^A]\dif u\right|\calF_t\right].
\end{equation}
Since the terminal price for both the European and American exchange options are the same, we have the relation $$\E_\Q\left[\left.e^{-r(T-t)}C_T^A\right|\calF_t\right] = \E_\Q\left[\left.e^{-r(T-t)}(S_{1,T}-S_{2,T})^+\right|\calF_t\right] = C_t^E,$$ and so therefore we have
\begin{equation}
\label{eqn-AmericanExcOptionRep2}
C_t^A = C_t^E - \E_\Q\left[\left.\int_t^T e^{-r(u-t)}\bar{\calL}[C_{u-}^A]\dif u\right|\calF_t\right].
\end{equation}

At this point, we introduce the early exercise region for American exchange options. Consolidating the findings of \citet{BroadieDetemple-1997}, \citet{Touzi-1999}, and \citet{CheangChiarella-2011}, we define the early exercise region (or stopping region) at time $t$ for the American exchange option to be given by
\begin{equation}
\calS = \left\{\left(S_{1,t},S_{2,t}\right)\in\mathbb{R}^2_+:S_{1,t}\geq B(v_{1,t},v_{2,t},t)S_{2,t}\right\},
\end{equation}
where the line
\begin{equation}
s_{1,t} = B(v_{1,t},v_{2,t},t)S_{2,t}
\end{equation}
on the $s_1s_2$-plane is called the early exercise boundary. The continuation region $\calC$, given by
\begin{equation}
\calC = \left\{\left(S_{1,t},S_{2,t}\right)\in\mathbb{R}^2_+:S_{1,t}< B(v_{1,t},v_{2,t},t)S_{2,t}\right\},
\end{equation}
is complement of $\calS$ in the first quadrant of the $s_1s_2$-plane (see Figure \ref{fig-EarlyExerciseBoundary}).\footnote{Optimal stopping arguments \citep[see for example][in the case of the single-asset American put option]{Myneni-1992} lead to a definition of the stopping and continuation regions in terms of the option price and its payoff. In the case of the American exchange option \citep{BroadieDetemple-1997}, the stopping region is given by $$\calS = \left\{\left(S_{1,t},S_{2,t}\right)\in\mathbb{R}^2_+:C_t^A(S_{1,t},S_{2,t},v_{1,t},v_{2,t}) = (S_{1,t}-S_{2,t})^+\right\}.$$ Analogously, the continuation region is defined as $$\calC = \left\{\left(S_{1,t},S_{2,t}\right)\in\mathbb{R}^2_+:C_t^A(S_{1,t},S_{2,t},v_{1,t},v_{2,t}) > (S_{1,t}-S_{2,t})^+\right\}.$$} Note that $B(v_{1,t},v_{2,t},t)$ (which is greater than or equal to 1) represents the critical price ratio of stocks 1 and 2 above which it is optimal to exercise the option. If at time $t$ the stock prices are in $\calS$, then it is optimal to exercise the American exchange option. If the stock prices are in $\calC$ at time $t$, then the option should not be exercised and the investor should ``continue'' to wait until it is optimal.

\begin{figure}
\centering
\includegraphics[width=3.5in]{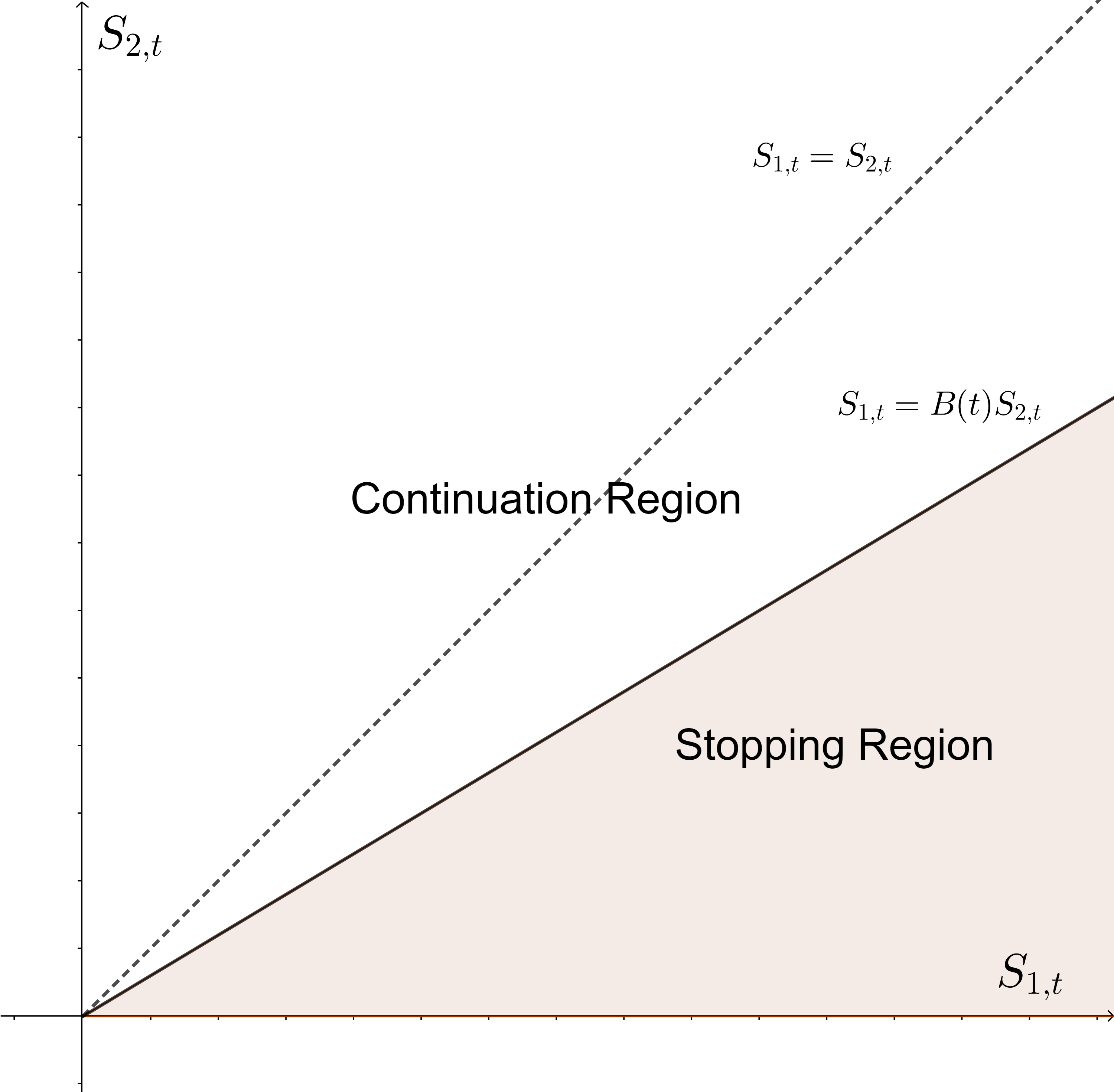}
\caption{The early exercise boundary and the continuation and stopping regions for the American exchange option, adapted from \citet{BroadieDetemple-1997, CheangChiarella-2011}. In the SVJD case, the slope $B(t)$ (shorthand for $B(v_{1,t},v_{2,t},t)$) of the early exercise boundary is also dependent on $v_{1,t}$ and $v_{2,t}$ \citep{Touzi-1999}.}
\label{fig-EarlyExerciseBoundary}
\end{figure}

In the stopping region $\calS$, the discounted American exchange option price is a strict supermartingale \citep{BroadieDetemple-1997, CheangChiarella-2011}, hence from equation \eqref{eqn-AmericanExcOptionRep1} it should hold that
\begin{equation}
\bar{\calL}[C_{t-}^A] < 0, \qquad \text{if $(S_{1,t-}, S_{2,t-})\in\calS$}.
\end{equation}
In the continuation region $\calC$, it is suboptimal to exercise the American option (i.e. the American option behaves like the European option), and so the discounted American option price is a martingale. This implies that
\begin{equation}
\bar{\calL}[C_{t-}^A] = 0, \qquad \text{if $(S_{1,t-}, S_{2,t-})\in\calC$}.
\end{equation}

From the last analysis above, it follows that the American exchange option price is the solution to the IPDE
\begin{align}
\begin{split}
\label{eqn-SVJD-AmExcOptionIPDE}
rC_{t-}^A
	& = \calL[C_{t-}^A]+\tilde{\lambda}_1\E_\Q^{Y_1}\left[C_t^A\left(S_{1,t-}e^{Y_1},S_{2,t-},v_{1,t},v_{2,t}\right)-C_{t-}^A\right]\\
	& \qquad + \tilde{\lambda}_2\E_\Q^{Y_2}\left[C_t^A\left(S_{1,t-}, S_{2,t-}e^{Y_2},v_{1,t},v_{2,t}\right)-C_{t-}^A\right],
\end{split}
\end{align}
where $0<S_{1,t}<B(v_{1,t},v_{2,t},t)S_{2,t}$, $S_{2,t}>0$, $v_{1,t}>0$, $v_{2,t}>0$, and $0\leq t<T$ (i.e. the continuation region $\calC$). Terminal and boundary conditions for the IPDE are given by
\begin{align}
\begin{split}
C_T^A(S_{1,T},S_{2,T},v_{1,T},v_{2,T}) & = (S_{1,T}-S_{2,T})^+\\
C_t^A(0,S_{2,t},v_{1,t},v_{2,t}) & = 0, \qquad S_{2,t}>0\\
C_t^A(S_{1,t},0,v_{1,t},v_{2,t}) & = S_{1,t}, \qquad S_{1,t}>0.
\end{split}
\end{align}
The IPDE is also supplemented by the value-matching condition
\begin{equation}
C_t^A = S_{1,t}-S_{2,t}, \qquad \text{for $S_{1,t}\geq B(v_{1,t},v_{2,t},t)S_{2,t}$ and $S_{2,t}>0$}
\end{equation}
that gives the price of the option once stock prices enter the stopping region $\calS$ \citep{CheangChiarella-2011, ChiarellaKangMeyer-2015}. In addition to the value-matching condition, we also require additional conditions on the behavior of the American exchange option price along the early exercise boundary. These additional conditions are known as smooth-pasting conditions:
\begin{align}
\begin{split}
\label{eqn-SVJD-SmoothPastingConditions}
\lim_{s_1\to B(v_1,v_2,t)s_2}\pder[C_t^A]{s_1}(s_1,s_2,v_1,v_2) & = 1\\
\lim_{s_1\to B(v_1,v_2,t)s_2}\pder[C_t^A]{s_2}(s_1,s_2,v_1,v_2) & = -1\\
\lim_{s_1\to B(v_1,v_2,t)s_2}\pder[C_t^A]{v_1}(s_1,s_2,v_1,v_2) & = 0\\
\lim_{s_1\to B(v_1,v_2,t)s_2}\pder[C_t^A]{v_2}(s_1,s_2,v_1,v_2) & = 0\\
\lim_{s_1\to B(v_1,v_2,t)s_2}\pder[C_t^A]{t}(s_1,s_2,v_1,v_2) & =0.
\end{split}
\end{align}
The smooth-pasting conditions result from the assumption that the holder of the American exchange option will maximize its value by selecting the appropriate exercise strategy. Consequently, the smooth-pasting conditions also ensure that the first-order partial derivatives of $C_t^A$ with respect to stock prices and volatility will be continuous for any value of $S_1$ and $S_2$ \citep{ChiarellaKangMeyer-2015}.

\begin{rem}
The American exchange option pricing problem can be formulated as a linear complementarity problem based on the preceding arguments. Since the stock price ordered pair $(S_{1,t-},S_{2,t-})$ is always either in the stopping or continuation region, it follows that
\begin{equation}
\label{eqn-LCPFormulation}
\min\left\{-\bar{\calL}\left[C_{t-}^A\right], C_{t-}^A - (S_{1,t-}-S_{2,t-})\right\} = 0
\end{equation}
for any $S_{1,t-}>0$, $S_{2,t-}>0$, and $0< t\leq T$. This equation is supplemented by the inequalities
\begin{align}
\bar{\calL}[C_{t-}^A] & \leq 0\\
C_{t-}^A & \geq S_{1,t-}-S_{2,t-},
\end{align}
both of which are true for any $S_{1,t-}>0$, $S_{2,t-}>0$, and $0< t\leq T$, and the terminal payoff condition
\begin{equation}
C_T^A = S_{1,T}-S_{2,T}.
\end{equation}
Note that this formulation does not contain the early exercise boundary, which may be useful in implementing numerical methods to find the solution to the pricing problem \citep{Seydel-2017}. Similar to what was noted in \citet[Chapter 7]{ChiarellaKangMeyer-2015}, the early exercise boundary may be found \emph{after} finding the solution as the boundary of the set $$\left\{S_{1,t},S_{2,t},v_{1,t},v_{2,t}:C^A_t(S_{1,t},S_{2,t},v_{1,t},v_{2,t})>S_{1,t}-S_{2,t}\right\}.$$
\end{rem}

To proceed with equation \eqref{eqn-AmericanExcOptionRep2}, we decompose the conditional expectation into integrals over the stopping and continuation regions. If $\calA_t$ denotes the event that the stock prices at time $t$ are in $\calS$, then
\begin{align*}
C_t^A
	& = C_t^E - \int_t^T e^{-r(u-t)}\E_\Q\left[\left.\bar{\calL}[C_{u-}^A]\vm{1}_{\calA_u}+\bar{\calL}[C_{u-}^A]\vm{1}_{\calA_u^c}\right|\calF_t\right]\dif u\\
	& = C_t^E - \int_t^T e^{-r(u-t)}\E_\Q\left[\left.\bar{\calL}[C_{u-}^A]\vm{1}_{\calA_u}\right|\calF_t\right]\dif u
\end{align*}
Suppose $(S_{1,t-},S_{2,t-})\in\calS$ (i.e. the event $\calA_t$ holds). It follows from the value-matching condition that $$C_{t-}^A = C_t^A(S_{1,t-},S_{2,t-},v_{1,t},v_{2,t}) = S_{1,t-}-S_{2,t-},$$ and so $$\pder[C_{t-}^A]{s_1} = 1 , \quad \pder[C_{t-}^A]{s_1} = -1, \quad \pder[C_{t-}^A]{t} = \pder[C_{t-}^A]{v_1} = \pder[C_{t-}^A]{v_2} = 0$$ for $(S_{1,t-},S_{2,t-})\in\calS$. Consequently, all second-order partial derivatives appearing in $\bar{\calL}[C_{t-}^A]$ vanish in the stopping region. This implies that
\begin{align*}
\bar{\calL}[C_{u-}^A]\vm{1}_{\calA_u}
	& = \calL[C_{u-}^A]\vm{1}_{\calA_u}-rC_{u-}^A\vm{1}_{\calA_u}\\
	& \qquad + \tilde{\lambda}_1\E_\Q^{Y_1}\left[C_u^A\left(S_{1,u-}e^{Y_1},S_{2,u-},v_{1,u},v_{2,u}\right)-C_{u-}^A\right]\vm{1}_{\calA_u}\\
	& \qquad + \tilde{\lambda}_2\E_\Q^{Y_2}\left[C_u^A\left(S_{1,u-},S_{2,u-}e^{Y_2},v_{1,u},v_{2,u}\right)-C_{u-}^A\right]\vm{1}_{\calA_u}\\
	& = \left[(r-q_1-\tilde{\lambda}_1\tilde{\kappa}_1)S_{1,u-}(1)+(r-q_2-\tilde{\lambda}_2\tilde{\kappa}_2)S_{2,u-}(-1)\right]\vm{1}_{\calA_u}\\
	& \qquad -r\left(S_{1,u-}-S_{2,u-}\right)\vm{1}_{\calA_u}\\
	& \qquad + \tilde{\lambda}_1\E_\Q^{Y_1}\left[C_u^A\left(S_{1,u-}e^{Y_1},S_{2,u-},v_{1,u},v_{2,u}\right)-(S_{1,u-}-S_{2,u-})\right]\vm{1}_{\calA_u}\\
	& \qquad + \tilde{\lambda}_2\E_\Q^{Y_2}\left[C_u^A\left(S_{1,u-},S_{2,u-}e^{Y_2},v_{1,u},v_{2,u}\right)-(S_{1,u-}-S_{2,u-})\right]\vm{1}_{\calA_u}.
\end{align*}
Recalling that $\tilde{\kappa}_i = \E_\Q^{Y_i}(e^{Y_i}-1)$, the last expression simplifies to
\begin{align*}
\bar{\calL}[C_{u-}^A]\vm{1}_{\calA_u}
	& = -\left[q_1 S_{1,u-}-q_2 S_{2,u-}\right]\vm{1}_{\calA_u}\\
	& \qquad + \tilde{\lambda}_1\E_\Q^{Y_1}\left[C_u^A\left(S_{1,u-}e^{Y_1},S_{2,u-},v_{1,u},v_{2,u}\right)-\left(S_{1,u-}e^{Y_1}-S_{2,u-}\right)\right]\vm{1}_{\calA_u}\\
	& \qquad + \tilde{\lambda}_2\E_\Q^{Y_2}\left[C_u^A\left(S_{1,u-},S_{2,u-}e^{Y_2},v_{1,u},v_{2,u}\right)-\left(S_{1,u-}-S_{2,u-}e^{Y_2}\right)\right]\vm{1}_{\calA_u}
\end{align*}
It follows that the American exchange option price is given by
{\scriptsize
\begin{align}
\begin{split}
C_t^A
	& = C_t^E + \int_t^T e^{-r(u-t)}\E_\Q\left[\left.\left(q_1 S_{1,u-}-q_2 S_{2,u-}\right)\vm{1}_{\calA_u}\right|\calF_t\right]\dif u\\
	& \qquad - \int_t^T e^{-r(u-t)}\E_\Q\left\{\left.\tilde{\lambda}_1\E_\Q^{Y_1}\left[C_u^A\left(S_{1,u-}e^{Y_1},S_{2,u-},v_{1,u},v_{2,u}\right)-\left(S_{1,u-}e^{Y_1}-S_{2,u-}\right)\right]\vm{1}_{\calA_u}\right|\calF_t\right\}\dif u\\
	& \qquad - \int_t^T e^{-r(u-t)}\E_\Q\left\{\left.\tilde{\lambda}_2\E_\Q^{Y_2}\left[C_u^A\left(S_{1,u-},S_{2,u-}e^{Y_2},v_{1,u},v_{2,u}\right)-\left(S_{1,u-}-S_{2,u-}e^{Y_2}\right)\right]\vm{1}_{\calA_u}\right|\calF_t\right\}\dif u
\end{split}
\end{align}}

Note that the discounted price of the American exchange option $\tilde{C}_t^A$ is the Snell envelope (the smallest supermartingale majorant) of the discounted intrinsic value $e^{-rt}(S_{1,t}-S_{2,t})^+$, following the optimal stopping arguments for American options \citep{Karatzas-1988, Myneni-1992, CheangChiarella-2011}. Therefore, the discounted American exchange option price is always greater than or equal to $e^{-rt}(S_{1,t}-S_{2,t})$, with equality occurring only in the stopping region. In this light, we note that $$C_u^A\left(S_{1,u-}e^{Y_1},S_{2,u-},v_{1,u},v_{2,u}\right)-\left(S_{1,u-}e^{Y_1}-S_{2,u-}\right) = 0$$ if $S_{1,u-}e^{Y_1}\geq B(v_{1,u},v_{2,u},u)S_{2,u-}$ (i.e. when $(S_{1,u-}e^{Y_1},S_{2,u-})\in\calS$) and $$C_u^A\left(S_{1,u-}e^{Y_1},S_{2,u-},v_{1,u},v_{2,u}\right)-\left(S_{1,u-}e^{Y_1}-S_{2,u-}\right) > 0$$ when $S_{1,u-}e^{Y_1}< B(v_{1,u},v_{2,u},u)S_{2,u-}$ (i.e. when $(S_{1,u-}e^{Y_1},S_{2,u-})\in\calC$). In the same manner, $$C_u^A\left(S_{1,u-},S_{2,u-}e^{Y_2},v_{1,u},v_{2,u}\right)-\left(S_{1,u-}-S_{2,u-}e^{Y_2}\right) = 0$$ when $S_{1,u-}\geq B(v_{1,u},v_{2,u},u)S_{2,u-}e^{Y_2}$ (i.e. when $(S_{1,u-},S_{2,u-}e^{Y_2})\in\calS$) and $$C_u^A\left(S_{1,u-},S_{2,u-}e^{Y_2},v_{1,u},v_{2,u}\right)-\left(S_{1,u-}-S_{2,u-}e^{Y_2}\right) > 0$$ when $S_{1,u-}< B(v_{1,u},v_{2,u},u)S_{2,u-}e^{Y_2}$ (i.e. when $(S_{1,u-},S_{2,u-}e^{Y_2})\in\calC$). In line with these observations, define the events $\calA_{1,u}$ and $\calA_{2,u}$ as
\begin{align*}
\calA_{1,u} & = \calA_u \cap \left\{S_{1,u-}e^{Y_1}< B(v_{1,u},v_{2,u},u)S_{2,u-}\right\} = \left\{B_u\leq\frac{S_{1,u-}}{S_{2,u-}}<B_u e^{-Y_1}\right\}\\
\calA_{2,u} & = \calA_u \cap \left\{S_{1,u-}< B(v_{1,u},v_{2,u},u)S_{2,u-}e^{Y_2}\right\} = \left\{B_u\leq\frac{S_{1,u-}}{S_{2,u-}}<B_u e^{Y_2}\right\},
\end{align*}
where $B_u$ is shorthand for $B(v_{1,u},v_{2,u},u)$. Following the arguments made above, the American exchange option price is therefore given by
{\scriptsize
\begin{align*}
C_t^A
	& = C_t^E + \int_t^T e^{-r(u-t)}\E_\Q\left[\left.\left(q_1 S_{1,u-}-q_2 S_{2,u-}\right)\vm{1}_{\calA_u}\right|\calF_t\right]\dif u\\
	& \qquad - \int_t^T e^{-r(u-t)}\E_\Q\left\{\left.\tilde{\lambda}_1\E_\Q^{Y_1}\left[C_u^A\left(S_{1,u-}e^{Y_1},S_{2,u-},v_{1,u},v_{2,u}\right)-\left(S_{1,u-}e^{Y_1}-S_{2,u-}\right)\right]\vm{1}_{\calA_{1,u}}\right|\calF_t\right\}\dif u\\
	& \qquad - \int_t^T e^{-r(u-t)}\E_\Q\left\{\left.\tilde{\lambda}_2\E_\Q^{Y_2}\left[C_u^A\left(S_{1,u-},S_{2,u-}e^{Y_2},v_{1,u},v_{2,u}\right)-\left(S_{1,u-}-S_{2,u-}e^{Y_2}\right)\right]\vm{1}_{\calA_{2,u}}\right|\calF_t\right\}\dif u
\end{align*}}

We summarize the results of the preceding calculations in the following proposition.

\begin{prop}
\label{prop-SVJD-AmExcOpPriceRep}
The price of the American exchange option admits the representation
\begin{equation}
\label{eqn-SVJD-AmExcOpPriceRep}
C_t^A(S_{1,t},S_{2,t},v_{1,t},v_{2,t}) = C_t^E(S_{1,t},S_{2,t},v_{1,t},v_{2,t})+C_t^P(S_{1,t},S_{2,t},v_{1,t},v_{2,t}),
\end{equation}
where $C_t^E$ is the price of the corresponding European exchange option and $C_t^P$ is the early exercise premium of the American exchange option. The early exercise premium is given by

{\scriptsize
\begin{align}
\begin{split}
C_t^P
	& = \int_t^T e^{-r(u-t)}\E_\Q\left[\left.\left(q_1 S_{1,u-}-q_2 S_{2,u-}\right)\vm{1}_{\calA_u}\right|\calF_t\right]\dif u\\
	& \qquad - \int_t^T e^{-r(u-t)}\E_\Q\left\{\left.\tilde{\lambda}_1\E_\Q^{Y_1}\left[C_u^A\left(S_{1,u-}e^{Y_1},S_{2,u-},v_{1,u},v_{2,u}\right)-\left(S_{1,u-}e^{Y_1}-S_{2,u-}\right)\right]\vm{1}_{\calA_{1,u}}\right|\calF_t\right\}\dif u\\
	& \qquad - \int_t^T e^{-r(u-t)}\E_\Q\left\{\left.\tilde{\lambda}_2\E_\Q^{Y_2}\left[C_u^A\left(S_{1,u-},S_{2,u-}e^{Y_2},v_{1,u},v_{2,u}\right)-\left(S_{1,u-}-S_{2,u-}e^{Y_2}\right)\right]\vm{1}_{\calA_{2,u}}\right|\calF_t\right\}\dif u,
\end{split}
\end{align}}

\noindent where the events $\calA_u$, $\calA_{1,u}$, and $\calA_{2,u}$ are defined as
\begin{align}
\calA_u & = \left\{(S_{1,u-},S_{2,u-})\in\calS\right\} = \{S_{1,u-}\geq B_u S_{2,u-}\}\\
\calA_{1,u} & = \calA_u \cap \left\{S_{1,u-}e^{Y_1}< B(v_{1,u},v_{2,u},u)S_{2,u-}\right\} = \left\{B_u\leq\frac{S_{1,u-}}{S_{2,u-}}<B_u e^{-Y_1}\right\}\\
\calA_{2,u} & = \calA_u \cap \left\{S_{1,u-}< B(v_{1,u},v_{2,u},u)S_{2,u-}e^{Y_2}\right\} = \left\{B_u\leq\frac{S_{1,u-}}{S_{2,u-}}<B_u e^{Y_2}\right\}.
\end{align}
Here, $B_u = B(v_{1,u},v_{2,u},u)$ is the early exercise boundary at time $u$.
\end{prop}

\begin{rem}
The Fourier inversion formulas we obtained for $C_t^E$ in equations \eqref{eqn-EuExcOp-FTPrice} and \eqref{eqn-EuExcOp-FTPrice-Decomposed-Timet} may be used as an input in equation \eqref{eqn-SVJD-AmExcOpPriceRep}.
\end{rem}

The event $\calA_{i,u}$, $i=1,2$, represents the event that the pre-jump stock prices $S_{1,u-}$ and $S_{2,u-}$ were in the stopping/early exercise region, but due to a jump in stock $i$ at time $u$, the post-jump stock prices were sent back to the continuation region. This interpretation is analogous to that offered in \citet{Pham-1997} for the single-asset jump-diffusion case and \citet{CheangChiarella-2011} for the case of exchange options under jump-diffusion dynamics.

The decomposition offered in the previous proposition is also similar to that in \citet{CheangChiarella-2011}. We also note that the early exercise premium can also be decomposed into a premium arising from the diffusion of the dynamics (the positive term) and rebalancing costs arising from the possibility that stock prices suddenly jump back into the continuation region (the negative terms), as was emphasized by \citet{Gukhal-2001} and \citet{CheangChiarella-2011}. In these computations, we note that the exercise boundary, and consequently the events defined with respect to the exercise boundary, are all dependent on the volatility levels.

From the value-matching condition, we note that
\begin{equation}
C_t^A(B_t S_{2,t},S_{2,t},v_{1,t},v_{2,t}) = S_{2,t}\left(B_t-1\right)
\end{equation}
when the stock prices are on the early exercise boundary (i.e. when $S_{1,t} = B_t S_{2,t}$). From the earlier proposition, we may therefore express the early exercise boundary $B_t = B(v_{1,t},v_{2,t},t)$ as a solution to the equation
\begin{equation}
S_{2,t}(B_t-1) = C_t^E(B_t S_{2,t},S_{2,t},v_{1,t},v_{2,t})+C_t^P(B_t S_{2,t},S_{2,t},v_{1,t},v_{2,t}).
\end{equation}
Note however that this equation must solved as a linked system in conjunction with equation \eqref{eqn-SVJD-AmExcOpPriceRep}, since the equation for the early exercise boundary involves the (yet unknown) American exchange option price $C_t^A$.

\section{Summary and Conclusions}
\label{sec-Conclusion}

In this paper, we have provided an extension to the results of \citet{Margrabe-1978} and \citet{CheangChiarella-2011} to consider the case where, aside from the presence of jumps, asset prices are also driven by a stochastic volatility process. To facilitate changes of measure from the objective probability measure to other equivalent measures, we introduced a Radon-Nikod\'ym derivative process, which requires some assumptions on the parameters of the volatility processes. From the construction of the Radon-Nikod\'ym derivative, it was noted that equivalent martingale measures are not unique, which therefore can lead to multiple plausible option prices. 

Representations for European exchange option prices were derived using two methods. The first method employs the change-of-num\'eraire technique to obtain a representation that is similar to the classical \citet{Margrabe-1978}. Alternatively, we considered additional assumptions on the model's correlation structure to allow an explicit form of the joint-characteristic function of the log-prices of the stocks. This, in turn, enabled us to represent European exchange option prices using in terms of this characteristic function using the results of \citet{CaldanaFusai-2013}. We were able to show that the European exchange option price obtained via the \citet{CaldanaFusai-2013} method can also be written in a form that is consistent with the characterization obtained via the change-of-num\'eraire procedure.

Finally, we demonstrated that the American exchange option price can also be represented as the sum of the price of the corresponding European exchange option price and an early exercise premium, similar to the findings of \citet{BroadieDetemple-1997}, \citet{Gukhal-2001}, and \citet{CheangChiarella-2011} in the case of jump-diffusion dynamics. We were also able to show that the early exercise premium can be decomposed into a premium on the diffusive component of asset prices and a premium owing to the possibility of jumps back into the continuation region right before exercise.

The use of a stochastic volatility jump-diffusion model for asset prices indeed allows for a more accurate characterization of asset prices but presents some complications in obtaining option prices. By making some minor additional assumptions on the correlation structure of the market model, we were able to obtain a representation of the European exchange option price in terms of Fourier inversion formulas. The representations we obtained may be numerically evaluated via Monte Carlo simulation or fast Fourier transform methods \citep[see][for example]{HurdZhou-2010, CaldanaFusai-2013, CaneOlivares-2014}. Meanwhile, extensions to the numerical methods proposed by \citet{Chiarella-2009} and \citet{ChiarellaZiveyi-2011} may be considered in providing a numerical solution for the American exchange option pricing problem. The efficacy of these methods in implementing our exchange option price representations is a topic for further research.

\section*{Acknowledgments}

The second author is supported by a Research Training Program International (RPTi) scholarship awarded by the Australian Commonwealth Government and by a Faculty Development Grant from the Loyola Schools of Ateneo de Manila University.

\bibliographystyle{tfcad}
\bibliography{references}


\appendix

\section{On the Martingale Property of the Stochastic Exponential}
\label{sec-app-MartingaleProperty}

Here, we show that the process $\{\calM_{i,t}\}$ defined by equation \eqref{eqn-StochExp} is a martingale under $\prob$. We first note that
\begin{align*}
\E_\prob\left[\exp\left\{-\lambda_i\kappa_i t+\sum_{n=1}^{N_{i,t}}Y_{i,n}\right\}\right]
	& = e^{-\lambda_i\kappa_i t}\E_\prob\left[\exp\left\{\sum_{n=1}^{N_{i,t}}Y_{i,n}\right\}\right]\\
	& = e^{-\lambda_i\kappa_i t}\cdot\exp\left\{\lambda_i t\left(\E_\prob(e^{Y_i})-1\right)\right\}\\
	& = e^{-\lambda_i\kappa_i t}e^{\lambda_i\kappa_i t} = 1.
\end{align*}
Next we examine if the process $$E_t=\exp\left\{-\frac{1}{2}\int_0^t v_{i,s}\dif s+\int_0^t \sqrt{v_{i,s}}\dif W_{i,s}\right\}$$ is a $\prob$-martingale. From \citet{Kuo-2006} Theorem 8.7.3 and \citet{WongHeyde-2004} Theorem 1, $E_t$ is a martingale if and only if $\E_\prob(E_t) = 1$ for all $t\in[0,T]$. \citet{Kuo-2006} notes that this condition is generally difficult to verify, but the stronger Novikov's condition, $\E_\prob[\exp\{\frac{1}{2}\int_0^T v_{i,t}\dif t\}]<\infty$, may be used instead.\footnote{As Novikov's condition is stronger as pointed out by \citet{Kuo-2006}, weaker alternative conditions for the martingale property are discussed in \citet{WongHeyde-2004}.}

Following Proposition 2.1 of \citet{AndersenPiterbarg-2007}, the process $\{v_{i,t}\}$ satisfying the conditions in Assumption \ref{asp-ParameterAssumptions} does not hit 0 and does not explode in in finite time \citep[see also][Chapter 9]{Lewis-2000}. Hence, $0<v_{i,t}<\infty$ almost surely for all $t\in[0,T]$, and so $\int_0^T v_{i,t}\dif t<\infty$ almost surely. Thus, Novikov's condition is satisfied, allowing us to conclude that $E_t$ is a martingale and that $\E_\prob(E_t)=1$.

\end{document}